\documentclass[11pt,reqno]{scrartcl}

\usepackage{amsmath,amssymb,amsthm,graphicx,bbm}
\usepackage{verbatim,xcolor}
\usepackage[shortlabels]{enumitem}
\setlist[enumerate]{itemsep=5pt}

\usepackage{dsfont,times}

\usepackage{tikz,tikz-cd,ifthen,pgfmath}
\usetikzlibrary{patterns}
\usetikzlibrary{calc}
\usetikzlibrary{arrows.meta}
\usetikzlibrary{decorations.markings}
\usetikzlibrary{decorations.pathmorphing}

\newcommand{\bb}[1]{\mathbbm{#1}}

\newcommand{\EE}{\bb{E}}

\newcommand{\tr}{\mathrm{tr}}

\newcommand{\Tv}{{\mathrm{TV}}}
\newcommand{\TV}{{\delta}}

\newcommand{\bra}[1]{\langle#1|}
\newcommand{\ket}[1]{|#1\rangle}
\newcommand{\braket}[2]{\langle#1 | #2\rangle}

\newcommand{\rank}{\mathrm{rank}}

\newcommand{\CC}{\mathbb{C}}

\newcommand{\RR}{\mathbb{R}}

\newcommand{\ZZ}{\mathbb{Z}}

\newcommand{\one}{\hbox{\rm 1\kern-.27em I}}

\newcommand{\be}{\begin{equation}}
\newcommand{\ee}{\end{equation}}

\newcommand{\BS}{\boldsymbol\sigma}    

\newcommand{\BE}{\boldsymbol\eta}    

\newtheorem{theorem}{Theorem}
\newtheorem{corollary}[theorem]{Corollary}
\newtheorem{lemma}[theorem]{Lemma}
\newtheorem{definition}[theorem]{Definition}
\newtheorem{proposition}[theorem]{Proposition}

\newtheorem{remark}[theorem]{Remark}

\numberwithin{equation}{section}
\numberwithin{theorem}{section}
\newcommand{\Z}{\mathbb{Z}}

\DeclareMathOperator{\sign}{sign}
\DeclareMathOperator{\Arg}{Arg}

\newcommand{\1}{\mathds{1}}        

\makeatletter

\let\Re\undefined

\DeclareMathOperator{\Re}{Re \,}

\begin{document}

	\addtokomafont{author}{\raggedright}		
\title{ \raggedright  
Entanglement entropy bounds for pure states of rapid decorrelation}



	\author{
	\begin{minipage}{.8\textwidth} 
	Michael Aizenman\textsuperscript{1}  and Simone Warzel\textsuperscript{2,3\,}\footnote{Corresponding author:  simone.warzel@tum.de}\\[1ex]
	\small \textsuperscript{1}Departments of Physics and Mathematics, Princeton University, Princeton, USA\\
	\textsuperscript{2}Departments of Mathematics and Physics, TU M\"unchen, Garching, Germany\\
	\textsuperscript{3}Munich Center for Quantum Science and Technology, Germany
	\end{minipage}
	}
	\date{\small March 14, 2025}
	
	\maketitle
	\vspace*{-1cm}

	\minisec{Abstract}
For pure states of multi-dimensional quantum lattice systems, which in a convenient computational basis have amplitude and phase structure of sufficiently rapid decorrelation,  we construct high fidelity approximations of relatively low complexity.   These are used for
a conditional proof of  area-law bounds for the states' entanglement entropy.  The condition is also shown to imply exponential decay of the state's mutual information between disjoint  regions, 
and hence exponential clustering of local observables. 
The applicability of the general results is demonstrated on the 
quantum Ising model in transverse field.  Combined with available model-specific information on spin-spin correlations, we establish an area-law type bound on the entanglement in the model's subcritical ground states, valid in all dimensions and up to the model's quantum phase transition.

\setcounter{tocdepth}{1}
{\small \tableofcontents}
\newpage
\section{Introduction} 
A striking feature of the quantum entropy, in which it  differs from its classical analog, 
is that in multicomponent systems, the entropy of a state  can be zero even when that of its restriction to a subsystem  is strictly positive.  
In such cases the entropy of the reduced state  is a measure of the quantum entanglement between the subsystem and the rest~\cite{OyPe93,Woud18}.  
One of our goals here is to present conditional decoupling criteria for an
area-law type bound  on the entanglement entropy of quantum states  in general space dimensions.    

\subsection{Entanglement entropy}
Our discussion is set in the context of quantum lattice systems whose single site components are qudits, of state space   associated with the complex Hilbert space $\CC^\nu$.
 The extended system's Hilbert space is the tensor-product $ \mathcal{H}_W := \otimes_{u \in W} \mathbb{C}^\nu $, indexed by sites of a finite subset $ W \subset \mathbb{Z}^d $.   For each choice of a basis of $\CC^\nu$, of $\nu$-elements whose collection we denote $\Omega_0$, there is a  corresponding orthonormal basis of $ \mathcal{H}_W$.  Its elements $ \ket{\BS} $ are labeled by    
  $ \boldsymbol{\sigma} = (\sigma_u)_{u\in W} \in \Omega_0^W=:\Omega_W$, 
configurations  of single-site variables with values in $ \Omega_0$ .\\

The composite system's pure states are rank-one projections 
associated with  normalised vectors $\ket{ \psi} \in \mathcal{H}_W$.  
In Dirac's notation these are  presentable as  $ | \psi \rangle\langle \psi | $.
More generally, quantum states are associated with density operators $\varrho$ which are non-negative  and of unit trace, 
each defining an  expectation-value functional  $O\mapsto \tr_{\mathcal H_W} O \varrho$  on the self-adjoint operators $ O $ on $ \mathcal{H}_W $ (which are the system's observables).

 Given $\ket \psi\in  \mathcal{H}_W$ and a subset  $ A\subset W $,  the reduced state  on $  \mathcal{H}_{A} $  is defined through the restriction of the expectation values $\bra \psi O \ket \psi$ to observables of the form $O= O_{A} \otimes \1_{W\backslash A} $.  Stated equivalently, the reduced state is the partial trace over the remaining subsystem, 
$ \varrho_{A}(\psi) := \tr_{\mathcal{H}_{W\backslash A} }  \ket \psi \bra \psi   $. 
A measure of $\ket \psi$'s  bipartite entanglement    is   $ \varrho_{A}(\psi) $'s von Neumann entropy:  
\begin{equation}\label{eq:vonNeumann}
S(\varrho_A) :=-  \tr_{ \mathcal{H}_{A}} \varrho_{A} \ln \varrho_{A}  = - \! \sum_{j=1}^{\dim \mathcal{H}_{A}}  \lambda_j \log \lambda_j \, ,   
\end{equation}
where $(\lambda_j)$ ranges over the eigenvalues of $\varrho_{A} $.

The entropy of any state in $A$ is trivially bounded by the set's  volume $|A|$: 
\be 
S(\varrho_A)  \leq  \ln \dim \mathcal{H}_{A} \propto |A| \, . 
\ee
Page's law \cite{Page:1993aa,Foong:1994aa} states that for 
vectors $ \ket{\psi}$ sampled uniformly  from the unit sphere of 
$ \mathcal{H}_W $  the entropy is typically of the order of this upper bound, down by only a multiplicative factor of order $O(1)$.
In contrast, the entanglement entropy of ground states  of  a local lattice Hamiltonian on  $  \mathcal{H}_W $ is expected to satisfy an \emph{area law}, i.e. $ S(\varrho_A) \propto  |\partial A | $, where $ \partial A := \{ u \in A \ | \ d(u,W\backslash A) = 1 \} $, with an enhancement at a quantum phase transition, cf.~\cite{Calabrese:2004aa,Eisert:2010aa}.

In this work, we shed light on this contrast through a decorrelation criterion under which 
the entanglement entropy of a given state's restriction to non-empty $A \subset \Z^d$  satisfies
\be \label{eq:area}
S(\varrho_A) \leq C \   l_0(A)\  | \partial A | \, ,  
\ee
at some $C <\infty$, which is independent of $A$, and a decoupling distance $ l_0(A) $, whose dependence on $ A $  relies on model-specific input.  
Section~\ref{sec:delta} elaborates on general conditions which ensure that the decoupling distance is bounded by a logarithm of the area, $  l_0(A)  \leq C \ln |\partial A | $.

  \subsection{Mutual information}\label{sec:mutualintro}

The mutual information  in a state $\ket{\psi}$ between two disjoint sets $ A_1, A_2 \subset W $  is 
\begin{eqnarray}\label{def:mutualinfo}
I_\psi(A_1 : A_2) := S(A_1) +  S(A_2)  - S( A_1\cup A_2)  =    S( A_1)  - S( A_1 \, |\,  A_2)\, . 
\end{eqnarray}
where $S(A) \equiv S(\varrho_A(\psi)) $ is an abbreviation for the entropy of the reduced state $\varrho_A(\psi)$, and  the last term is the conditional entropy 
$S( A_1 \,|\, A_2)  := S( A_1\cup A_2) - S(A_2)$.  Viewed yet differently, the mutual information encodes the relative entropy of $ \varrho_A(\psi) $ relative to the product state $ \varrho_{A_1}(\psi) \otimes \varrho_{A_2}(\psi) $.  \\

Our approach to the entanglement entropy, which proceeds through entropy reduction in a comparison state, also allows to conclude that under similar decorrelation assumptions the mutual information between two disjoint cubes $ A_1, A_2 $ decays exponentially in their distance:
\begin{equation} \label{mutual_info_bound}  
	I_\psi(A_1 : A_2)  \leq C \max\left\{ |\partial A_1|, |\partial A_2 | \right\}^\kappa \ \exp\left(- d(A_1,A_2)/\eta \right) \, ,
\end{equation}
at some $C,  \kappa , \eta \in (0,\infty) $.  
As a corollary we learn that  the state $\ket{\psi}$ exhibits \emph{exponential clustering}.  The two statements are linked through the quantum Pinsker inequality~\cite{HOT81,Woud18,Wolf:2008aa}, which bounds the covariance of  any 
pair of observables $ O_{A_j} $, acting on their $ \mathcal{H}_{A_j} $, 
 in terms of the mutual information of that state:
\begin{equation}\label{eq:Pinsker}
  \left|   \bra{\psi} O_{A_1} \otimes  O_{A_2} \ket{\psi} -  \bra{\psi} O_{A_1} \ket{\psi}\bra{\psi}   O_{A_2} \ket{\psi}  \right| 
  \leq \| O_{A_1} \| \| O_{A_1} \| \sqrt{2(\ln 2 )\ \ I_\psi(A_1 : A_2)} \, ,
\end{equation}
where the norm $ \| \cdot \| $  in the right side is the operator norm. 

\subsection{States of  quasi-local structure in a computational basis}  

The enabling criteria presented below for entropy reduction 
refer to the structure of a state $\ket{\psi}\in \mathcal{H}_W  $ as it appears in a convenient \emph{computational basis}   
in which it is presentable  as 
\be\label{eq:defpsip}
\ket{\psi}  = \sum_{\BS\in \Omega_W}e^{i\theta(\BS)}\, \sqrt{p(\BS) } \, \,  \ket{\BS} 
\ee
 in terms of  
a probability measure on $ \Omega_W $  of weights $ p(\BS) := |\psi(\BS)|^2  $,   
with rapid decay of correlations, and 
a phase  function $\theta:  \Omega_W \to (-\pi,\pi] $, $ \theta(\BS):=\Arg \psi(\BS) $ of  similar properties. \\

The discussion is not limited to but it simplifies in the case of states for which $\theta(\BS) \equiv 0$.  Such states are referred to as ``sign-problem free'' or \emph{stoquastic} - a term invoking the relevance of a probabilistic perspective.  Among  the examples of stoquasticity are the ground states of Hamiltonians $ H $ for which $e^{-\beta H}$ 
is positivity-preserving in the specified computational  basis.   A sufficient condition for that is that  the off-diagonal terms of $(-H)$ are real and non-negative.

Various applications of  states' probabilistic structure, e.g. conditions for symmetry breaking, or slow decay of correlations,  appeared  in studies of  specific models~\cite{Toth:1993aa,AKN94,Aizenman:1994aa,Grimmett:2008aa,Bjornberg:2009aa,Crawford:2010aa,Uel13,Bjornberg:2015kn,Aizenman:2020aa}.   
Implications of stoquasticity have also been investigated from the complexity point of view~\cite{BDOT08,Bravyi:2009aa,Klassen2019}.   

For stoquastic states our criterion will be expressed in terms of the  decorrelation rate 
of the classical looking probability  $ p(\BS) $.  An example to which this trivially applies, but to which our criterion is not limited, are  
states 
of classical Gibbs probability distribution under
additive local interactions, of the form 
\be\label{eq:Gibbs1}
p(\BS) \propto \exp\left( \sum_{u,v\in W} J_{u,v} \ g(T_u \BS ) \ g( T_v\BS)  \right) \,.
\ee
Here $J_{u,v}$  is a kernel of rapid enough decay, 
$ g: \Omega_{\Z^d} \to \mathbb{R} $ a bounded local function, and $ T_u:  \Omega_{\Z^d} \to  \Omega_{\Z^d} $ the translation by $ u $.   In case the kernel is finite range these states are examples of tensor-network states.  Our criterion  holds in that case  irrespective of phase transitions,  due to the classical state's  Markov property.

The general results presented here are not restricted to stoquasticity.  They will include a condition which requires the phase function be  decomposable into   sums of  
essentially local contributions.  That is a limitation since  even in the class of  Hamiltonians build of commuting terms there are ground states of interest with more intricate phase structure, cf.~\cite{Hast16}.  On the constructive side:  
phase functions  for which the assumed condition is met include functions of the form
\be \label{add_phase}
\theta(\BS) = \sum_{u,v\in W} \widehat J_{u,v} \ \widehat g(T_u \BS ) \  \widehat  g( T_v\BS)  \, 
\ee 
with $\widehat J_{u,v}$  a kernel of rapid enough decay, and $ \widehat  g: \Omega_{\Z^d} \to \mathbb{R} $ a bounded local function. 

\subsection{Example: the quantum Ising model}

Among the states to which our general discussion applies are the ground states of the quantum  (or  transversal field)  Ising model (QIM).   
The system's Hilbert space is $\mathcal H_W = \otimes_{u\in W}\CC^{2}$,  
and the Hamiltonian (at $h^z\equiv 0$) is
\be \label{QIM_H}
H =  - \sum_{u, v \in W} J_{u-v} \ S^z_u S^z_v -  b \sum_{u\in W}    S^x_u  \, , 
\ee 
 written in terms of elements of the Pauli triples of spin operators $ (S_u^x, S_u^y, S_u^z)$  
which act  on the $u$-component of the tensor product.  The coupling $ J_{u-v} \geq 0 $ is assumed to be finite range and ferromagnetic.  A convenient computational basis in  which the model's thermal and ground state are stoquastic is the eigenbasis of the local Pauli $ S^z $-operators.  Another such basis is the one in which the local $ S^x $ are diagonal.

The system's ground state undergoes a quantum phase transition at a threshold value~$b_c > 0$, which depends on $d$ and the couplings $J$.  The transition is manifested in the behavior of the model's two point correlation function, and most clearly in its infinite volume limit
\be   \label{eq:corr_lim}
\langle S^z_u S^z_v \rangle^{(\ZZ^d)} := \lim_{L\to \infty} \langle S^z_u S^z_v \rangle^{(W_L)} \, . 
\ee 
The superscript on the correlation indicates the domain over which the model is formulated, with $W_L = [-L,L]^d \cap \mathbb{Z}^d $.   In these terms, the model's ground state   undergoes the following transition (cf.~\cite{suzuki2012quantum, Bjornberg:2009aa}):\\
\noindent $\bullet$  For $ b > b_c $,  its  ground state correlation function decay exponentially:
\be \label{eq:shapness} 
0\leq \langle S^z_u S^z_v \rangle^{(\ZZ^d )} 
\leq C \ \exp\left(-|u-v|/\xi(b)\right)\, \quad \mbox{ at some $\xi(b) \in (0, \infty)$}\, . 
\ee 
\noindent $\bullet$  For $b \in [0, b_c)$ the correlation function exhibits long range order: 
\be
\langle S^z_u S^z_v \rangle^{(\ZZ^d )}  \geq M(b)^2 >0
\ee 
In that case the limit is decomposable into a superposition of two state functionals, 
$\langle O \rangle^{(\ZZ^d)}  = \frac 12 \big[ 
\langle O \rangle^{(\ZZ^d, +)} + 
\langle O \rangle^{(\ZZ^d, -)} 
   \big] 
   $
with $ \langle  S^z_u \rangle^{(\ZZ^d, \pm)} = \pm M(b)$, 
each with exponential decay of truncated correlations: 
\begin{align}
0\leq \langle  S^z_u ;  S^z_v \rangle^{(\ZZ^d, \pm)}  & :=  
\langle S^z_u S^z_v\rangle^{(\ZZ^d, \pm)} - \langle S^z_u \rangle^{(\ZZ^d, \pm)} \langle S^z_v\rangle^{(\ZZ^d, \pm)} \notag \\
& \leq  C \exp\left(-|u-v|/\xi(b)\right) \, . 
\end{align}

\noindent $\bullet$  At $ b = b_c $: $M(b)=0$~\cite{Bjornberg:2015kn}, but unlike at $b>b_c$ the  spin-spin correlation decays by a dimension-dependent power law. 

We shall not delve here into the direct formulation of the model's infinite-volume version (cf.~\cite{BratRob87}), in which the  $(\pm)$ states are presentable as ground states of an infinite spin array.  Instead, we shall focus on bounds for finite volumes $W_L$ which hold uniformly in $L$.  Thus, unless specified explicitly, the states refer to the model formulated in a finite subset $W\subset \ZZ^d$, which will often be omitted in the notation.  

We  also refrain from the more complete discussion of boundary conditions.  Unless specified otherwise there are to be assumed here as either free or constant (in terms of $\sigma^z$).  
For each of these,  the Perron-Frobenius theorem applied to $ \exp\left(-\beta H\right)$ in the $\sigma^z$-computational basis allows to conclude that the model's finite volume ground states are  unique at any $ b \neq 0 $.

The general results developed below, combined with  known decay properties of the $ S^z $-$ S^z $-correlations in this model \cite{Bjornberg:2009aa, Ding:2023aa}  enable us to prove  the following.
\begin{theorem}[Bounds for QIM] \label{thm:EEQIM}
 For the quantum Ising model on $W \subset \Z^d$, $d\geq 1$, with the  Hamiltonian \eqref{QIM_H}  at $b > b_c(d) $ (i.e. throughout the ground states' subcritical regime):
 \begin{enumerate}
 \item The entanglement entropy of the ground state's restriction to non-empty rectangular domains $A\subset \Z^d$ satisfies:
 \be \label{eq:areaQIM}
S(\varrho_A) \leq C \     | \partial A |  \ \ln |\partial A | \, ,  
\ee
at some $C <\infty$. 
 \item The ground state's mutual information between disjoint rectangular domains decays exponentially, satisfying 
 \eqref{mutual_info_bound}.  
 \item 
The ground state is exponentially clustering.  In particular  its truncated $ x $-correlations decay exponentially:  there is $ C , \xi(b) \in (0,\infty) $ such that for any $ W $, and $ u,v \in W $:
\be\label{eq:xexp}
 \left| \langle S^x_u ; S^x_v \rangle^{(W)} \right| \leq C \   \exp\left(-|u-v|/\xi(b)\right) \, .  
\ee
\end{enumerate}
\end{theorem} 

It should be noted that for $ d = 1 $ where $\ln |\partial A|$  is  constant,  \eqref{eq:areaQIM} yields a simple area law.  \\

The proof of Theorem~\ref{thm:EEQIM} is presented in Section~\ref{sec:QIM}.  The argument suggests that~\eqref{eq:areaQIM} may be valid also for the supercritical phase.  However, to confirm that our general criterion is   applicable also to $b<b_c$  
would require improvements in the currently available model-specific information.\\

\noindent{\textbf{Relations to previous results specific to QIM:} } 
The previous results on the entanglement specific to  this model 
were restricted to one dimension and there to significantly sub-critical states~\cite{Grimmett:2008aa,Grimmett:2020aa,Campanino:2020aa}.   
The broader Theorem~\ref{thm:EEQIM} addresses the entire subcritical regime, in all dimensions.
The difference is traced below to a key handicap in the approach previously taken. 

The bound on the mutual information and exponential clustering~\eqref{eq:xexp} extends a result of Crawford and Ioffe~\cite{Crawford:2010aa} that was derived using the model's random current representation.  It could potentially also be obtained using such methods.

\section{State approximations of decreased complexity}

A  known approach  towards bounds on the entanglement entropy $ S(\varrho_A(\psi)) $ of a specified pure state $ \ket{\psi} $  is to seek high-fidelity approximations of the state $\varrho_A(\psi)$ in $A$ in terms of states of low-complexity, e.g.~states of low rank -- comparable with the exponential of the reduced state's expected entropy~\cite{AKLV13,Aharonov:2014aa,AnshuArad22}.    Seeking such approximations, we present here two versions of approximations of $\varrho_A(\psi)$  invoking quantum versions of  conditioning on the state in a buffer $B$ of adjustable width,  enveloping the set $A$ and separating it from the much bigger $C=  W\setminus (A\sqcup B)$, cf.~Figure~\ref{fig:1}.  The first is expressed directly in terms of the reduced state $\varrho_A(\psi) $. The second is expressed in terms of  a Markovian approximation to the given state $ \ket{\psi} $.   

\begin{figure}[htb]
\begin{center}\includegraphics[width=.5\textwidth]{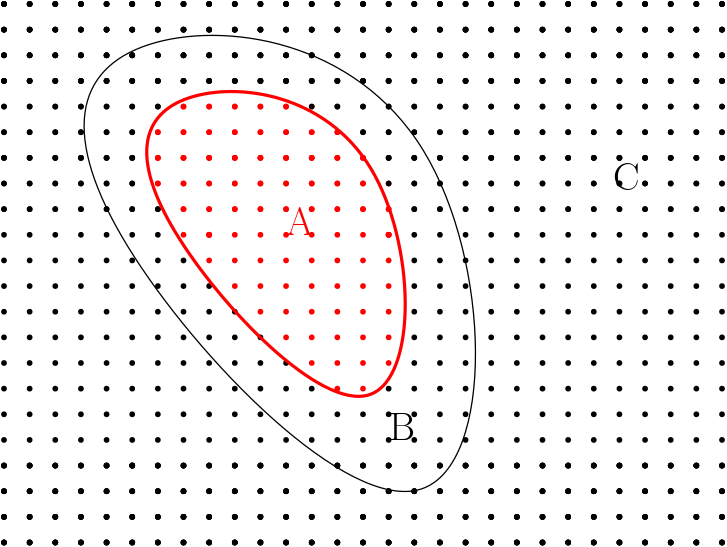}
\caption{Sketch of the decomposition of $ W \subset \mathbb{Z}^d $ into $ A $, a buffer $ B $, and the remaining set $ C $.}\label{fig:1} 
\end{center}\end{figure}

\subsection{Reduced state approximation}

Given a state vector $\ket \psi \in \mathcal H_W$ and a tripartition $W = A \sqcup B \sqcup C$ in which $B$ is a buffer separating $A$ from the larger region $C$, let $\{\ket{ \BS_B}\}$ be an  orthonormal basis for $\mathcal H_B$, e.g. labeled by spin configurations $ \BS_B \in \Omega_B $ from a convenient computational basis.
Let $\ket{\psi(\BS_B)} $ denote the normalized vector defined by 
\be 
\big[ \bb1_A \otimes  \ket{\BS_B} \bra{\BS_B} \otimes \bb1_C  \big]\, \ket {\psi}  =: \sqrt{p(\BS_B)}   \, \ket{\psi(\BS_B)}
\ee 
with $p(\BS_B) \in [0,1]$.  
In these terms, the state vector is $\ket{\psi} = \sum_{\BS_B}  \sqrt{p(\BS_B)}   \,\,  \ket{\psi(\BS_B)}$, and its  reduced state in $A$  can be presented as
\be  \label{mixed2} 
\varrho_A( \psi) = \;   \sum_{\BS_B} p(\BS_B) \  \varrho_A\left( \psi(\BS_B) \right)  
\quad\mbox{with}\quad   \varrho_A\left( \psi(\BS_B) \right) :=  \tr_{ \mathcal H_{B\sqcup C}}   \ket{\psi(\BS_B)}  \bra{\psi(\BS_B)} \, . \;
 \ee
To interpret this procedure, one may note that the restriction
of the reduced state of $\ket\psi$  in $ A $ is not changed if at first the state is pinched in $B$,
that is transformed  
\be \label{pinch}
 \ket\psi \bra \psi   \rightarrow  \sum_{\BS_B} p(\BS_B) \ \ket{\psi(\BS_B)} \bra{\psi(\BS_B)}   \, . 
\ee
The pinching, which can be viewed as the result of non-destructive measurement of $\BS_B$,
allows one to apply to the quantum systems the classical notion of conditioning.

The next observation is that while the rank of $  \varrho_A( \psi) $  generically reaches the lower of the dimensions $ \dim \mathcal H_A$ and $\dim \mathcal H_{BC} $, the density matrices  $\varrho_A\left( \psi(\BS_B)\right)$ may potentially be approximated by lower rank operators, being conditioned on specified states of the buffer. 
In particular, an approximation to $  \varrho_A( \psi)  $ of rank bounded by $\dim \mathcal H_{B} $  is given by
\begin{equation} \label{eq:reduced}
 \boxed{  \widehat{\varrho}_A\left(\psi\right) := \sum_{\BS_B} p(\BS_B) \, \ket{ \widehat{\phi}_A(\BS_B)}  \bra{ \widehat{\phi}_A(\BS_B)} \, , }
\end{equation}
where ${ \widehat \phi}_A(\BS_B) $ is a normalized eigenvector corresponding to the largest eigenvalue, $ \lambda(\BS_B) := \left\|  \varrho_A\left(\psi(\BS_B)\right) \right\|  $, of the density matrix in~\eqref{mixed2}. 

For the fidelity of this approximation, as measured by the trace-norm ($  \| \cdot  \|_1 $) distance, one has the  bound 
\begin{equation}\label{eq:genfidestimate}
 \left\| \varrho_A\left(\psi\right) - \widehat{\varrho}_A\left(\psi\right) \right\|_1 \leq 2 \ \sum_{\BS_B} p(\BS_B) \, \left( 1 - \lambda(\BS_B)  \right)  \, , 
\end{equation}
which follows from straightforward applications of the triangle inequality. 
Note that $ 1- \lambda(\BS_B) =  1- \left\|  \varrho_A\left(\psi(\BS_B)\right) \right\| $ measures the proximity of $ \varrho_A\left( \psi(\BS_B) \right) $ to a rank-one operator. If $ \ket{\psi(\BS_B)} \in \mathcal H_A\otimes \mathcal H_{BC} $ factorizes into a product state, then $ \lambda(\BS_B)  = 1 $.   
One may also explore more general approximations of the reduced state $\varrho_A\left( \psi \right)$ through 
sums over more than one of the leading  terms in  the  spectral  decomposition of  $\varrho_A\left( \psi(\BS_B) \right)$.\\



\subsection{Markovian state approximation}
For another approach, one may start from a convenient computational orthonormal basis over $\mathcal H_A\otimes \mathcal H_B\otimes \mathcal H_C$, indexed by spin configurations $(\BS_A,\BS_B,\BS_C)\in \Omega_A \times \Omega_B \times \Omega_C$, which represent the eigenvalues of commuting observables. 
Their joint probability distribution under a given state vector $ \ket{\psi} $ is 
$ 
p(\BS)  = |\psi(\BS)|^2 
$ 
which can  be written as 
\be 
p(\BS)  = 
 p(\BS_B) \, p(\BS_C | \BS_B) \, p(\BS_A | \BS_B, \BS_C)
\ee 
 in terms of the induced conditional probabilities  $p(\BS_A, \BS_C | \BS_B) = p(\BS_A, \BS_B, \BS_C) / p(\BS_B) $ and 
 $p(\BS_A | \BS_B) = \sum_{\BS_C} p(\BS_A, \BS_C | \BS_B) $.
%
Seeking a Markovian approximation of $\ket{\psi}$ it is natural to consider states $ \sum_{\BS}\sqrt{\widehat p(\BS) }\, e^{i \widehat\theta(\BS) } \ket{\BS} $ with the  
modified probability distribution 
\be
\mkern-100mu\mathrm{(I)} \mkern80mu  \widehat p(\BS_A, \BS_B,  \BS_C)  =  p(\BS_B) \  p(\BS_A \,|\,  \BS_B)  \ p(\BS_C \,|\,  \BS_B) \, , 
\ee
under which $\BS_A$ and $\BS_C$ retain their marginal distribution but are conditionally independent, 
and phase functions  of the form
\be
\mkern-160mu \mathrm{(II)}  \mkern80mu  \widehat\theta(\BS_A, \BS_B, \BS_C)  = \alpha(\BS_A, \BS_B)  + \gamma(\BS_C , \BS_B)
\ee
 with some  $\alpha: \Omega_A \times \Omega_B \to \mathbb{R} $ and $\gamma:  \Omega_C  \times \Omega_B \to \mathbb{R}$ which  may be variationally selected. 
With such a choice a markovian approximation takes the form
\be \label{eq:appr}
\boxed{\ket{\psi_{(B)}} := \sum_{\BS}\sqrt{p(\BS_B) }\, \ket{\phi_A(\BS_B)}\otimes \ket{\BS_B} \otimes \ket{\phi_C(\BS_B)} }
\ee 
with normalized vectors
\begin{align}\label{def:phivectots}
 \ket{\phi_A(\BS_B)} &:=  \sum_{\BS_A}  e^{i\alpha(\BS_A , \BS_B)} \sqrt{p(\BS_A \,|\, \BS_B)}\,   \ket{\BS_A}  \\
	  \ket{\phi_C(\BS_B)} &:=  \sum_{\BS_C}  e^{i\gamma(\BS_C , \BS_B)} \sqrt{p(\BS_C \,|\, \BS_B)}\,   \ket{\BS_C}  \, .  \notag
\end{align}
For the vector~\eqref{eq:appr} the state has the Markov property in the sense that  
conditioned on $\BS_B$ it factorizes into a product state over $\mathcal H_A \otimes \mathcal H_C$.
Its reduced state on $A$,
 \begin{equation}\label{reduced}
  \varrho_A(\psi_{(B)})=  \sum_{\BS_B} p(\BS_B) \  
 \ket{\phi_A(\BS_B)} \bra{\phi_A(\BS_B)} \, ,
 \end{equation}
has a rank bounded by $\dim \mathcal H_{B} $. 
Hence, in contrast to $\varrho_A(\psi)$,  the induced state   $\varrho_A(\psi_{(B)})$ is guaranteed to satisfy
\be \label{eq:compest} 
S\big(\varrho_A(\psi_{(B)})\big) \leq \ln\big( \nu^{|B|}\big) \propto  l \, | \ \partial A|  \, .
\ee

\noindent \textbf{Remark:}~~For concreteness sake, it may be noted that a simple way to split an additive phase as a sum of the above form is to  pick somewhat arbitrarily a configuration $\BE \in \Omega_W$ and set $\alpha$ and $\gamma$ as
\begin{equation}
\begin{split}
\alpha(\BS_A , \BS_B) &= \theta(\BS_A,\BS_B,\BE_C) - \theta(\BE_A,\BS_B,\BE_C)  \\
\gamma(\BS_C , \BS_B) &=  \theta(\BE_A,\BS_B,\BS_C)  \, . 
 \end{split}
\end{equation}
In a purely additive case, as in \eqref{add_phase}, the sum $\alpha+\gamma$ depends on $\BE$ only through the terms which link directly to both $\BS_A$ and $\BS_C$.  In the absence of such,  this injection of $\BE$ is for convenience only, its selection having no effect on the analysis.

\section{Fidelity bound for the Markovian state approximation} \label{sec:fidelity}

Quantification of the deviation of $\ket{\psi}$ from the idealized conditions listed above as (I) and (II), which hold for $\ket{\psi_{(B)}}$, will involve two terms. \\

\noindent (I) \quad
To address the degree of correlations in $p(\BS)$ between $\BS_A$ and $\BS_C$, when both are conditioned on the configuration $ \BS_B $ of the buffer, we use the following terminology and notation. 

Without  yet conditioning on $B$, we denote by $ \TV(A \, | \, C) $  the variational difference between  the probability distribution of  $(\BS_A, \BS_C)$  and the independent product of the 
corresponding marginals:
 \begin{align} \label{def_Delta}
 \TV(A \, | \, C) &:= \frac 12 \sum_{\BS_A, \BS_C} \big| p(\BS_A, \BS_C) - p(\BS_A) p(\BS_C) \big|  \\ 
 &=  \sum_{\BS_C} p(\BS_C) \sum_{\BS_A}\big[ p(\BS_A | \BS_C) - p(\BS_A) \big]_+  \notag    \, . 
 \end{align}    
with $\big[ \cdot \big]_+$ denoting the positive part.  
Of more direct relevance for our analysis is  the following averaged conditional version of the above:
\begin{align}\label{def_cond_B}
\TV_B(A \, | \,  C)  
&:= \sum_{\BS_B, \BS_C} p(\BS_B,\BS_C)  \sum_{ \boldsymbol{\sigma}_A}  \left[  p(\boldsymbol{\sigma}_A | \BS_C \BS_B)- p(\boldsymbol{\sigma}_A  | \boldsymbol{\sigma}_B) \right]_+ \!   \notag \\ 
&\equiv \sum_{\BS_B} p(\BS_B) \, \TV_{\BS_B}(A \, | \, C)    
\end{align}
where $\TV_{\BS_B}$ refers to the $\TV(A \, | \,  C)$-measure of correlation in the joint distribution of $(\BS_A,  \BS_C)$ conditioned on the specified value of $\BS_B$. 

These definitions extend naturally to mutually disjoint $ A, B , C \subset W $, whose union does not exhaust $ W $. 
Section~\ref{sec:delta} includes a more detailed discussion of this quantity  as well as estimates in specific models.   One may briefly note that for classical Gibbs equilibrium measures~\eqref{eq:Gibbs1}  under local interactions  $\TV_B(A\, | \, C) = 0$ once the width of the buffer $ B $  exceed the interaction's range.   That  does not apply to quantum ground states whose Hamiltonians include non-commuting terms, even in the stoquastic case.   
Still, one may expect that in non-critical states, the buffer-conditioned correlations  would be weak, at least on average.\\

\noindent (II) \quad  For  an estimate on the degree of deviation from the second condition  
we denote 
\begin{equation}\label{def:K}
 \vartheta_B(A|C) :=    \inf_{\alpha, \gamma} \ \left \{  \sum_{\BS} p(\BS)   \left| 1 -  U_{\alpha,\gamma}(\BS)  \right|^2 \right\}^{\frac{1}{2}}    \, , 
 \end{equation}
where  
$$
 U_{\alpha,\gamma}(\BS) := e^{-i\theta(\BS)}
\exp\left( i \left[ \alpha(\BS_A , \BS_B) \ +\ \gamma(\BS_C , \BS_B) \right] \right) 
$$
and the optimization is over  functions $ \alpha : \Omega_A \times \Omega_B \to[0,2\pi] $ and 
$\gamma : \Omega_C  \times \Omega_B \to [0,2\pi]$. \\
Since $ \Omega_A , \Omega_B , \Omega_C $ are finite sets, $ [0,2\pi] $ is compact and  the right side of~\eqref{def:K} is continuous in the values of these functions, the infimum is a minimum. Hence, there exist optimizing functions $ \alpha_{AB}, \gamma_{AC} $ for the infimum.

Let us note that for states $ \ket{\psi} $, whose phase is an additive function of the form \eqref{add_phase}, the term $\vartheta_B(A | C)$ vanishes  if the buffer's width exceeds the range of the kernel $\widehat J $.  More generally, $\vartheta_B(A | C)$ should be small in case the phase $\theta(\BS) $ is presentable as a sum of quasi local terms similar to~\eqref{add_phase}, once the buffer's width exceeds the relevant coherence length by a factor of order $ \ln |\partial A|$.  \\

It will be useful to consider also an extension of the above to 
the case $ A $ is a disjoint union of  connected components $ A_1, \dots , A_n $,  separated by disjoint buffers $ B_1, \dots , B_n $ from  $C := W\setminus (A \sqcup  B)$ (at $B = \sqcup_{j=1}^n B_j $), as depicted in~Figure~\ref{fig:2} for $ n = 2 $.  In such case, we denote
\begin{equation}\label{eq:K2}
 \vartheta_{B_1,\dots B_n}(A_1, \dots, A_n  |C) :=    \inf_{\alpha_1 , \dots , \alpha_n, \gamma} \ \left \{  \sum_{\BS} p(\BS)   \left| 1 -  U_{(\alpha_j),\gamma}(\BS)  \right|^2 \right\}^{\frac{1}{2}}    \, , 
\end{equation} 
where the infimum is now over $ \alpha_j : \Omega_{A_j} \times \Omega_{B_j} \to[0,2\pi] $, $ j \in \{1,\dots, n \} $, and 
$\gamma : \Omega_C  \times \Omega_B \to [0,2\pi]$, and
$$
U_{(\alpha_j),\gamma}(\BS) :=  e^{-i\theta(\BS)}
\exp\left( i \left[ \sum_{j=1}^n\alpha_j(\BS_{A_j} , \BS_{B_j}) \ +\ \gamma(\BS_C , \BS_B) \right] \right) \, . 
$$
Since optimizers $ \alpha_1,\dots , \alpha_n ,\gamma $ of~\eqref{eq:K2}, yield bounds by choosing $ \alpha =\sum_{j=1}^n  \alpha_j $ and $ \gamma $ in~\eqref{def:K}, we have   $ \vartheta_B(A|C) \leq  \vartheta_{B_1,\dots B_n}(A_1, \dots, A_n  |C) $. \\

The above quantities will next be used to estimate the degree to which $\varrho_A(\psi_{(B)})$ approximates the state of our original interest $\varrho_A(\psi)$.  Their trace-norm distance relates to the fidelity as measured by the progenitor state's global overlap (somewhat in the spirit of Uhlmann's bound~\cite{Uhlmann:1970}):
\begin{align}\label{eq:FvG} 
\| \varrho_A(\psi) - \varrho_A(\psi_{(B)}) \|_1 & = \sup_{\| O_A \| \leq 1 } \left| \bra{\psi} O_A \otimes \bb{1} \ket{\psi} - \bra{\psi_{(B)}} O_A \otimes \bb{1} \ket{\psi_{(B)}} \right|  \notag \\
& \leq 2 \left\| \psi - \psi_{(B)} \right\| \leq 2 \sqrt{ 2 ( 1 - \Re \braket{\psi}{\psi_{(B)}} ) } \, . 
\end{align}
%
%
Our main estimate of the fidelity of the approximation of $\rho_A(\psi)$ by $\rho_A(\psi_{(B)})$ at different choices of the buffer $B$ is the following result.
\begin{theorem}[Fidelity bound] \label{thm:fidcor} 
For any  state vector $\ket{\psi}$ and its  approximation $\ket{\psi_{(B)}}$ defined by~\eqref{eq:appr} with optimizers $ \alpha = \alpha_{AB} $, $ \gamma = \gamma_{CB} $ from~\eqref{def:K}: 
\begin{equation}\label{eq:fidcor}
\Big[\tfrac{1}{2} \| \varrho_A(\psi) - \varrho_A(\psi_{(B)}) \|_1\Big]^2 \leq 2  \left|1- \braket{\psi}{\psi_{(B)}} \right| \leq 
2\   \TV_B(A \, | \,  C)  + 2\ \vartheta_B(A|C) \, . 
\end{equation}
\end{theorem}
\begin{proof} 
By  \eqref{eq:FvG} it suffices to prove the second inequality.  
Expressing  $\psi(\BS) = \sqrt{p(\BS)} \exp\left(i \theta(\BS)\right) $ and $ p(\BS)  = p(\BS_A | \BS_C \BS_B  ) \,\,p(\boldsymbol{\sigma}_C  | \boldsymbol{\sigma}_B)\,\,  p(\boldsymbol{\sigma}_B)$  in terms of conditional probabilities, we have 
\begin{align}
1- \braket{\psi}{\psi_{(B)}}  = & \ \sum_{\BS} p(\BS_B, \BS_C ) \sqrt{p(\BS_A|\BS_B,\BS_C)} \notag  \\
& \qquad \times   \left(  \sqrt{p(\BS_A|\BS_B,\BS_C)}  - U_{\alpha,\gamma}(\BS)\sqrt{p(\BS_A | \BS_B)} \right) \, .
\end{align}
 Writing $ U_{\alpha,\gamma} = 1 + (U_{\alpha,\gamma}-1) $, we split the sum into two terms. 
The absolute value of the part, which is free of $ U_{\alpha,\gamma} $, is 
 \begin{multline}
 \Big| \sum_{\BS} p(\BS_B, \BS_C ) \sqrt{p(\BS_A|\BS_B,\BS_C)} \left( \sqrt{p(\BS_A|\BS_B,\BS_C)}  - \sqrt{p(\BS_A | \BS_B)} \right) \Big| \\
 \leq   \frac{1}{2} \sum_{\BS} p(\BS_B, \BS_C ) \left| p(\BS_A|\BS_B,\BS_C)  - p(\BS_A | \BS_B) \right|   = \delta_B(A | C) \, .  
 \end{multline}
 The inequality is based on the estimates $ - [ b - a ]_- \leq  \sqrt{b} \ (\sqrt{b} - \sqrt{a} ) \leq [ b - a ]_+$ for  $ a,b \geq 0 $, together with the fact that the positive and negative part of the sum are equal.  This yields the first term in the right side of~\eqref{eq:fidcor}.
 
 The absolute value of the second part is estimated with the help of the triangle inequality and the Cauchy-Schwarz inequality for the sum:
  \begin{align}
  & \Big|  \sum_{\BS} p(\BS_B, \BS_C ) \sqrt{p(\BS_A|\BS_B,\BS_C) \ p(\BS_A | \BS_B)}  \left( U_{\alpha,\gamma}(\BS) - 1 \right) \Big| \notag \\
  & \leq  \sum_{\BS} p(\BS_B, \BS_C ) \sqrt{p(\BS_A|\BS_B,\BS_C) \ p(\BS_A | \BS_B)}  \left| U_{\alpha,\gamma}(\BS) - 1 \right| \notag \\
  & \leq  \Big( \sum_{\BS} p(\BS)  \left| U_{\alpha,\gamma}(\BS) - 1 \right|^2\Big)^{1/2} \Big( \sum_{\BS} p(\BS_B, \BS_C ) p(\BS_A | \BS_B)   \Big)^{1/2} \ = \ \vartheta_B(A|C) \,.
 \end{align}
This yields the second term on the right side of~\eqref{eq:fidcor}. 
 \end{proof}

\noindent \textbf{Remark:}~~Alternatively, one could have arrived at~\eqref{eq:fidcor} by starting from~\eqref{eq:genfidestimate} using the bound $p(\BS_B)\ \lambda(\BS_B) \geq \left|  \braket{\phi_A(\BS_B) \otimes \BS_B \otimes \phi_C(\BS_B)}{ \psi}  \right|^2 $ with the vectors from~\eqref{def:phivectots}. Applying the estimates in the proof of  Theorem~\ref{thm:fidcor} to the scalar product, one then concludes the outer estimate in~\eqref{eq:fidcor}.

\section{Entropic implications of rapid decorrelation}\label{sec:entropy}

 In this section, we relate the fidelity lower bounds formulated above with the following entropy estimates: 
 \begin{enumerate}[i.]
 \item An area-law bound corrected by $\ln |A|$ on the entanglement entropy in pure states of rapid decorrelation in  dimensions $d\geq 1$,  
 \item\label{label2}  An improved estimate in which  the log correction is replaced by $\ln |\partial A|$.  
\item \label{label3} A related upper bound for the difference in entanglement entropy between $ \ket{\psi} $ and $\ket{\psi_{(B)}} $.
\end{enumerate} 
The difference between $\ln |A|$ and $\ln |\partial A|$ accomplished in~\ref{label2}  is of particular value for one dimension, where the latter is just a constant.  It also affects~\ref{label3}.  

Throughout the discussion 
we   restrict attention to sets $ A $ which are regular in the following sense.
\begin{definition}\label{def:regularA}
We call $ A \subset W $  \emph{regular}  if the volume of a buffer $ B_l := \{ u \in W \backslash A  \ | \ d(u,A) \leq l \} $ of width  $ 1 \leq l \leq  L(A) :=  | A|\slash (c_d\ |\partial A| )$ is bounded by the surface area of $ A $ times $l$, i.e.
\begin{equation}\label{eq:Buffervol}
c_d  \  l \  | \partial A | \leq   \left|  B_l \right| \leq C_d \  l \  | \partial A | 
\end{equation}
with some dimension dependent $c_d, C_d \in (0,\infty) $.
\end{definition}

\subsection{An entanglement bound based on a single buffer estimate}

In the single scale version of our bound, we approximate the state $\varrho_A(\psi)$ by $\varrho_A(\psi_{(B)})$, that is, the reduced state of the \emph {conditionally decorrelated vector}~\eqref{eq:appr}. 
The entropy of $\varrho_A(\psi_{(B)})$ is upper bounded by through its Schmidt norm: 
\be  \label{eq:Schmidt}
S(\varrho_A(\psi_{(B)})) \leq |B|  \ln \nu \, .
\ee 

A bound on the difference of the two entropies is  provided by the following variant of Fannes' continuity bound~\cite{Fannes:1973aa,Audenaert:2007aa}.  
 \begin{equation}\label{eq:Fannes}
\left| S(\varrho_A) - S(\widehat \varrho_A)\right| \leq  \tfrac{1}{2} \left\| \varrho_A -  \widehat\varrho_A \right\|_1  \left(1+  \ln \frac{2 \ \rank (\varrho_A - \widehat \varrho_A) }{ \left\| \varrho_A -  \widehat\varrho_A \right\|_1 }  \right) \, . 
\end{equation}
For the reader's convenience, a proof is included here in Appendix~\ref{App:Fannes}. 

Under this relation, closeness in entropy requires proximity of the states in the trace norm to a degree which is affected by rank  considerations.   
For concreteness sake let us consider the case where the following condition is met.   (An example of that are the ground states of the quantum Ising model throughout its subcritical phase).  
\begin{definition} \label{def:expdecay1}
A state  $ \ket{\psi}  $ is said to exhibit \emph{exponentially fast  conditional decoupling}, at correlation length $\xi > 0$, if for any regular $ A \subset W $ and buffer $ B_l = \{ u \in W \backslash A  \ | \ d(u,A) \leq l \} $  of width $l$:
 \begin{equation}\label{ass:expdec0}
 \TV_{B_l}(A \, |\, C_l )  + \vartheta_{B_l}(A |  C_l)  \leq   C_\xi \ L(A)^{2\alpha} \ e^{-l/\xi}  
\end{equation}
with  $C_l= W\setminus (A\cup B_l)$, $ C_\xi , \alpha\in (0,\infty) $ and $L(A)$ the length of $A$ from Definition~\ref{def:regularA}. 
\end{definition} 

In that case, the combination of \eqref{eq:Schmidt} and \eqref{eq:Fannes} yields: 

\begin{theorem} \label{thm:singlescale}
For a  state $ \ket{\psi} $, with exponentially fast  conditional decoupling, the entanglement entropy corresponding to any regular $ A \subset W $ with $L(A)\geq 2$ satisfies
\be \label{eq:single}   
S(\varrho_A(\psi)) \leq  C_{\xi, d } \,  |\partial A|    \ln L(A)
\ee 
at some constat $ C_{\xi, d } < \infty $ which is independent of $ A $. 
\end{theorem} 
\begin{proof}
Aiming for an approximation of relatively low rank and high fidelity, let us set the buffer's width to  
$ l = 2(1+\alpha') \xi \ln L (A) $
with  $\alpha' > \alpha$.  
For this choice 
$
S(\varrho_A(\psi_{(B_l)})) \  \leq \  |B_l|   \log \nu  \leq   C_d\ l\ |\partial A| \ln \nu \,  $, 
by~\eqref{eq:Buffervol} and thus the claimed bound holds for the approximating state. 

Arranging the terms of \eqref{eq:Fannes} by order of magnitude, we get 
\begin{eqnarray} 
S(\varrho_A(\psi)) -  S(\varrho_A(\psi_{(B_l)}))  \ \leq \    \tfrac{1 }{2} \Delta_l^{(1)}   |A|  \ln \nu   +  O(1)   , 
\end{eqnarray}
with  $ \Delta_l^{(1)} := \left\| \varrho_A(\Psi) -  \varrho_A(\Psi_B)  \right\|_1 $
and the term $O(1)$ is the collection of terms 
$ \tfrac{1}{2}\Delta_l^{(1)}  \big[  1+ \ln\big( 2 /\Delta_l^{(1)} \big) \big] $.
Under  the assumption  \eqref{ass:expdec0} the bound \eqref{eq:fidcor} 
yields
\be 
\Delta_l^{(1)}   \ |A|  \leq   C_\xi^{1/2} c_d \ L(A)^{\alpha -\alpha'}   |\partial A| 
\ee
which is negligible for $\alpha' > \alpha $.  Hence, under the stated condition  the difference in entropies is insignificant. 
\end{proof} 

Reviewing the argument one may note that the above bound is suboptimal in that measuring the entropy difference in terms of $\| \varrho_A(\psi) -  \varrho_A(\psi_{(B)})  \|_1 $ ignores the fact that the states' disparity may be significant mainly in areas close to the boundary of $A$.  

We next derive an improved upper bound in which $\ln |A|$ is replaced by $\ln |\partial A|$,  derived through a multiscale analysis.   The correction is particularly significant for one dimension, in which case $ |\partial A|$ is of order $O(1)$ and hence the  log term in \eqref{eq:single} may be dropped. 

\subsection{An improved multiscale analysis}

Inspecting  \eqref{eq:vonNeumann}  one may note that a state's entropy $S(\rho_A)$ 
may be strongly affected by the eigenvalues   at the low end of $\varrho_A$'s  spectrum.  
Taking a cue from~\cite{Aharonov:2011aa, Brandao:2015ab, AnshuArad22}, the following entanglement estimate is derived by decomposing the entropy $S(\varrho_A)$ into a sum of contributions from different spectral regions, and employing buffers $ B_l $ of increasing width $ l \in \mathbb{N} $  
 for probing the contributions from the spectrum's lower scales.  
 
The analysis is conveniently quantified by the  distribution-type function
\begin{equation}\label{eq:distribution}
 \tau_{\varrho_A}(N) := \sum_{j > N} \lambda^\downarrow_j(\varrho_A) \, ,
\end{equation}
where the sum is over the eigenvalues of $\varrho_A$ labeled  in decreasing order.
Of particular relevance is the rate at which $\tau_{\varrho_A}(N)$  vanishes as $N$ increases toward
the dimension of the relevant Hilbert space. 
The following main lemma
estimates this quantity using 
the fidelity estimate~\eqref{eq:fidcor}. 
\begin{lemma}  Given a state $\ket{\psi}$,  for any $A\subset W$ and a buffer set $B$,   
 the spectral distribution function of $\varrho_A \equiv \varrho_A(\psi)$  satisfies 
\begin{equation}\label{eq:massTV}
\tau_{\varrho_A}\big(\nu^{|B|} \big) \leq 2 \ \TV_{B}(A \, |\, W \backslash ( A \sqcup B) ) + 2 \ \vartheta_B(A \, |\, W \backslash ( A \sqcup B) )  \,.
\end{equation}
\end{lemma} 
\begin{proof}
We use the variational characterization of 
\begin{equation}\label{eq:tauvar}
 \tau_{\varrho_A}(N)  = 1 - \max\left\{ \tr \varrho_A P_N \ | \ P_N \, \mbox{orthogonal projection of $ \rank \ P_N \leq N $} \right\}  
\end{equation} 
in the case $ N = \nu^{|B|} $ with $ P_N $ is the spectral projection onto the span of the vectors $ \ket{\phi_A(\BS_B)} $ with $ \BS_B \in \Omega_B $, cf.~\eqref{def:phivectots}. 
Since $ P_N \otimes \mathbbm{1}_{BC} \ket{ \psi_{(B)}}   = \ket{ \psi_{(B)}} $ and $ \mathbbm{1} \geq  \ket{ \psi_{(B)}}  \bra{ \psi_{(B)}}  $, we have
\be
\tr \varrho_A(\psi)  P_N = \langle \psi | P_N \otimes \mathbbm{1}_{BC} | \psi \rangle \geq \left| \braket{\psi}{\psi_{(B)}} \right|^2 . 
\ee
The stated inequality then follows by estimating $ 1 -  \left| \braket{\psi}{\psi_{(B)}} \right|^2 \leq 2  \left| 1- \braket{\psi}{\psi_{(B)}} \right| $ 
and the fidelity bound of Theorem~\ref{thm:fidcor}. 
\end{proof}

The following generalizes Definition~\ref{def:expdecay1}.

\begin{definition}\label{def:expdec}
Given a decaying function  $ \varphi: \mathbb{N} \to [0,1] $ 
we say that  a state $ \ket{\psi} $  is \emph{conditionally $ \varphi $-decoupled at rate $\varphi(l)$  beyond distance  $ l_0: \{ \mbox{subsets of $ W $} \} \to [0,\infty) $,} 
 if for any regular $ A \subset W $ and any $l\geq l_0(A) $:
 \begin{equation}\label{ass:expdec}
 \TV_{B_l}(A \, |\, C_l )  + \vartheta_{B_l}(A |  C_l)  \leq   \varphi(l-l_0(A))  
\end{equation}
with $ B_l = \{ u \in W \backslash A  \ | \ d(u,A) \leq l \} $ and $C_l= W\setminus (A\cup B_l)$. 
\end{definition} 
Two rate functions of natural interest are:
\begin{itemize} 
\item exponential:  $\varphi(k) = e^{-k/\xi}$, at some $ \xi > 0 $,
\item power law:  $ \varphi(k) =  \left( 1+ k/\xi \right)^{-\alpha } $ 
at some $ \alpha > 2 $.
\end{itemize} 
The distance function $ l_0 $ plays the role of the length scale at which the $ \varphi $-decay sets in and allows to absorb $ A $-dependent prefactors as in~\eqref{ass:expdec0}. These prefactors could be independent of $ A $, in which case $ l_0 $ is constant. \\

In this setting, we obtain the following general entanglement bound.    

\begin{theorem}[General bound]\label{thm:arealaw}
For a  state $ \ket{\psi} $, which is conditionally $\varphi $-decoupled beyond distance $ l_0 $, the entanglement entropy corresponding to any regular $ A \subset W $ is bounded according to
\begin{equation}\label{eq:arealaw}
S(\varrho_A(\psi)) \leq C_d \ | \partial A | \left( \ln \nu \right) \left[ l_0(A) +  2 \ (1+l_0(A)) \ I_1(0;A)  \right] +  2 I_2(0;A)  ,
\end{equation}
where
\begin{equation}\label{eq:defPHi}
 I_1(l;A) := \mkern-10mu \sum_{k = l }^{L(A)} \mkern-5mu  \varphi(k) (1+k) , \quad I_2(l;A) := \mkern-10mu  \sum_{k = l}^{L(A)}    \mkern-5mu \varphi(k)  \left( 1 + \ln \varphi(k)^{-1}  \right),
\end{equation}
with  $ L(A) :=  | A|\slash (c_d\ |\partial A| )$ and  $c_d, C_d $ are the constants from~\eqref{eq:Buffervol}. 
\end{theorem} 
\begin{proof}
We split the sum in~\eqref{eq:vonNeumann}  into contributions stemming from the summation index in the intervals 
$$ J_0 :=\big[1,\nu^{|B_{l_0}|} \big] \,  \mbox{and} \,  J_k := \big( \nu^{|B_{l_0+k-1}|}, \nu^{|B_{l_0+k}|} \big] \,  \mbox{with $ k \geq 1 $.}
$$
The bound~\eqref{eq:massTV} and the assumed decay~\eqref{ass:expdec} is used to control 
\begin{equation}\label{eq:muest}
\mu_k := \sum_{j \in J_k }  \lambda^\downarrow_j(\varrho_A) \leq \tau_{\varrho_A}\big(\nu^{|B_{l_0+k-1}|}\big) \leq 2 \  \varphi(k-1) 
\end{equation}
in case $ k \geq 1 $.  This will be employed to estimate 
\begin{align}\label{eq:Sk1}
	S_k \ &  :=   \sum_{j \in J_k }  \lambda^\downarrow_j(\varrho_A) \ln  \lambda^\downarrow_j(\varrho_A)^{-1} \leq \mu_k \ln\frac{|J_k|}{\mu_k}  \leq  \mu_k  \ln\frac{\nu^{|B_{l_0+k}|} }{\mu_k} . \end{align}
 For $ k = 0 $, we use the trivial bound $ \mu_0 := \sum_{j \in J_k }  \lambda^\downarrow_j(\varrho_A) \leq 1 $, which yields $ S_0 \leq |B_{l_0}| \ln \nu $. 
By~\eqref{eq:Buffervol}, this leads to the first term in the right of \eqref{eq:arealaw}.
If $ k \geq 1 $, we employ~\eqref{eq:muest} and \eqref{eq:Sk1} to further  bound
$$
S_k \leq 2 \ \varphi(k-1) \left( \ln \frac{ \nu^{|B_{l_0+k}|}}{ \varphi(k-1) } +1 \right) ,
$$
by the monotonicity of  $ F(x) = x \left(  \ln (a / x ) + 1 \right) $ for  $0 \leq x \leq 1 \leq a$. Collecting the contributions to
$   S(\varrho_A) = S_0 + \sum_{k\geq 1 }  S_k  $ and noticing that $l_0 + k $ does not exceed $ L(A) $, since the total number of eigenvalues of $ \varrho_A \equiv \varrho_A(\psi) $ is $ \nu^{|A|} $,  yields the result. 
\end{proof} 

In order to elucidate this general statement, let us add some comments and formulate a simple consequence.
 In case $ A $ is a cube, it is regular and $ L(A) $ is proportional to the cube's edge length.   Moreover, 
power law decay of $ \varphi $ with $ \alpha > 2 $ is enough to ensure the convergence of the series corresponding to the sums in~\eqref{eq:defPHi}. One may therefore upper bound these sums independently of~$ A$.
We thus arrive at the following, which is one of the main results of this work. 
\begin{corollary}[Area-law bound] \label{thm:gen_area1}
For any state $ \ket{\psi} $ which is conditionally $ \varphi $-decoupled with power law decay with $ \alpha > 2 $ beyond distance $ l_0 $, there is $ C  \in (0,\infty) $, independent of $ l_0 $,  such that 
the state's entanglement entropies   satisfy
\begin{equation}\label{eq:arealaw1}
S(\varrho_A(\psi)) \leq C\ | \partial A | \ ( 1+ l_0(A) )\, 
\end{equation}
for all  regular $ A  \subset W $.
\end{corollary} 
\begin{proof}
Under the assumption of fast algebraic decay with $ \alpha > 2 $, one has 
$$ I_1(0;A) \leq \sum_{k=0}^\infty  \frac{1+k}{(1+ k/\xi)^{\alpha}}  < \infty, \qquad   I_2(0;A)  \leq  \sum_{k=0}^\infty  \frac{1  + \alpha \ln (1+ k/\xi)}{(1+ k/\xi)^{\alpha} } < \infty \, .
$$ 
The claim then follows from Theorem~\ref{thm:arealaw}, with the additive constants  absorbed into the multiplicative constant  in the right side of~\eqref{eq:arealaw1}. 
\end{proof}

Several remarks are in order:
\begin{enumerate}
\item The bound~\eqref{eq:arealaw1} describes a strict area law in case the decoupling distance $l_0(A) $ is independent of $ A $. This applies in particular to stoquastic states with $ p $ any (even critical) classical Gibbs probability measure corresponding to a  finite range interaction, cf.~\eqref{eq:Gibbs1}. Such states fall into the category of tensor-network states, for which the area-law of the entanglement entropy is well known (cf.~\cite{Verstraete:2006sf,Wolf:2008aa}).  

\item For ground-states of the quantum Ising model existing techniques yield $ l_0(A) \propto \ln |\partial A | $ for general regular $ A $, which results in a logarithmically corrected area law unless $ d = 1 $.

\item The general results presented here have some partial overlap with theorems that were formulated for different classes of states, under different assumptions and restrictions:

\begin{enumerate}
\item Various facets of the area-law conjecture for  ground states were settled for $ d = 1  $ as well as for exactly solvable or non-interacting systems (cf.~\cite{Eisert:2010aa} for a review).  In a trailblazing work, Hastings~\cite{Hastings:2007aa}  proved it for ground states of one dimensional systems with gapped local Hamiltonians. 
The argument was streamlined in the work of Brand{\~a}o and  Horodecki~\cite{Brandao:2015ab}, which clarified that the argument is, in essence, based on the state's exponential clustering property.   	

\item 
For $ d = 2 $, under the assumption of frustration-freeness of  a gapped local Hamiltonian on $ W \subset \mathbb{Z}^2 $,  Anshu, Arad and Vidick~\cite{Anshu:2016aa} devised a technique, based on a so-called detectability argument, for constructing an approximate ground state projector.  Using it, Anshu, Arad and Gosset~\cite{AnshuArad22} established an area-law bound with a poly-logarithmic correction for the entanglement entropy of slub subsets $A$ of rectangular sets $ W $. 

\item 
Beyond the above two-dimensional case, only conditional results with 
assumptions that do not cover the full non-critical regime  are available~\cite{Masanes:2009aa,Cho:2014aa,Brandao:2015aa}.  The area law is also known to fail for some non-regular graphs~\cite{Aharonov:2014aa}.  

\end{enumerate}
 \item An alternative measure of entanglement is provided by  the  R\'enyi entropy $ S_\alpha(\varrho) := (1-\alpha)^{-1} \ln \tr \varrho^{\alpha} $.   Since it is monotone decreasing in $ \alpha \in ( 0,\infty) $ and converges to the von Neumann entropy for $ \alpha \to  1 $,  the above-discussed area-law bounds extend to all values $ \alpha \geq 1 $.  In particular, that applies to integer values $ \alpha \in \{2,3,\dots\}$.  
 
 \item For $ d = 1$  and nearly-critical ground states,  Cardy and Calabrese  \cite{Calabrese:2004aa} predicted the asymptotic value of the R\'enyi $2$ entanglement entropy,  expressed in terms of the correlation length and the central charge of the conformal field theory describing the critical ground state.
 \end{enumerate}
   
\subsection{Entropy approximations}

Arguments used in the proof of Theorem~\ref{thm:arealaw} also yield the following bound on entropy differences between the state of interest $ \ket{\psi} $ and its approximation $ \ket{\psi_{(B_l)}} $ corresponding to a buffer of width $ l $.  
\begin{theorem}[General entropy difference]\label{thm:compS}
In the situation of Theorem~\ref{thm:arealaw} for any $l \geq l_0(A) $:
\begin{align}\label{eq:contentgen}
\left| S(\varrho_A(\psi)) - S(\varrho_A(\psi_{(B_{l})})) \right| \leq &  \sqrt{2 \varphi(l-l_0(A))} \left( 1+ \ln \frac{\nu^{|B_{l}|}}{\sqrt{2\varphi(l-l_0(A))}}  \right) \\
& + 2 C_d | \partial A |\left( \ln \nu \right)  ( l  + 1) \  I_1(l-l_0(A);A)  +  2 I_2(l-l_0(A);A)   \, ,   \notag
\end{align}
with $ \ket{\psi_{(B_l)} }$ as defined in~\eqref{eq:appr} with $ B \equiv B_l $, $ C \equiv W \backslash ( A \sqcup B_l) $ and minimizers $ \alpha = \alpha_{AB} $, $ \gamma = \gamma_{C B} $ from~\eqref{def:K}. 
\end{theorem}
\begin{proof} We start as in the proof of Theorem~\ref{thm:arealaw} with $ l_0 $ substituted by $ l $. 
It was proven there that the contributions to  $ S(\varrho_A(\psi)) $ from intervals $ J_k =  \big( \nu^{|B_{l+k-1}|}, \nu^{|B_{l+k}|} \big] $ with $ k \geq 1 $ are bounded according to
$$ \sum_{k\geq 1} S_k \leq  2 C_d | \partial A |  \left( \ln \nu \right) ( l  + 1) \  I_1(l-l_0(A);A)  + 2 I_2(l-l_0(A);A)  \, . 
$$
The proof then proceeds by estimating the difference of $S( \widehat\varrho_A )$ with  $ \widehat\varrho_A \equiv \varrho_A(\psi_{(B_{l})}) $ and $ S_0 = - \sum_{j= 1}^{\nu^{|B_l|}} \lambda_j^\downarrow \ln  \lambda_j^\downarrow $, where $\big( \lambda_j^\downarrow \big)$ stand for the eigenvalues of $ \varrho_A \equiv \varrho_A(\psi) $  labeled  in decreasing order. For that purpose, let $ \mu_j^\downarrow $, $ j \in \{ 1, \dots , \nu^{|B_l|} \} $,  
denote the decreasing eigenvalues of $ \widehat\varrho_A $.  Then
\begin{align}\label{eq:Deltabound}
\frac{1}{2} \Delta_l  & := \frac{1}{2} \sum_{j=1}^{ \nu^{|B_{l}|}}  \left| \lambda_j^\downarrow -  \mu_j^\downarrow\right| 
\leq  \frac{1}{2} \left\| \varrho_A-  \widehat\varrho_A\right\|_1 \leq \sqrt{2 ( 1 - \Re\braket{\psi}{\psi_{(B_{l})}} }  \notag \\
&\,  \leq  \sqrt{2 \TV_{ B_l }(A \ | W \backslash ( A \sqcup B_l) ) } \leq \sqrt{2\varphi(l-l_0)} \, ,
\end{align} 
where we used \cite[Lemma IV.3.2]{Bathia}, the variational characterization of the trace norm, $ \| a \|_1 = \sup_{ \| b \| = 1}\left|  \tr  (a b) \right|  $, as well as~\eqref{eq:FvG}. The last line is from the fidelity bound~\eqref{eq:fidcor} and the assumed decay. 
The Fannes-type entropy bound from Proposition~\ref{cor:Fannes} thus yields
\begin{align}\label{eq:firstterm}
& \left| S_0 -  S(\widehat\varrho_A ) \right| =  \Big|  \sum_{j=1}^{\nu^{|B_l|}} \big(\lambda_j^\downarrow \ln  \lambda_j^\downarrow  - \mu_j^\downarrow \ln  \mu_j^\downarrow\big) \Big|   \leq \frac{1}{2} \Delta_l  \left(1+  \ln \frac{2\,   \nu^{|B_l|} }{ \Delta_l}  \right)  \, .
\end{align}
Inserting~\eqref{eq:Deltabound} into the right side and using the monotonicity of the functions then yields the first term  in the right side of~\eqref{eq:contentgen}. 
\end{proof}
Let us also put this result into context. 
Applying the trivial bound~\eqref{eq:compest} on the entropy of the comparison state,
which holds for arbitrary $ l \in \mathbb{N} $ and follows from the rank-estimate on $\varrho_A(\psi_{(B)}) $,  Theorem~\ref{thm:compS} essentially implies \eqref{eq:arealaw}. 
Since the rank of $ \varrho_A(\psi) $ is unknown, the bound~\eqref{eq:contentgen} does not immediately follow from the Fannes-Audenaert-continuity bounds~\cite{Fannes:1973aa,Audenaert:2007aa} or any of its relatives~\cite{Alicki:2004aa,Shirokov:2023aa}. \\

In the subsequent analysis of the mutual information, we need a minor modification of Theorem~\ref{thm:compS} involving the approximation of $ \psi $ by the vector $  \ket{\psi_{(B_l)} }$ from~\eqref{eq:appr} with phases now optimized according to~\eqref{eq:K2}. For simplicity, we restrict attention to the case $ n = 2 $, i.e., the phases will be adapted to the decomposition of $ A $ into two separated components $ A_1 , A_2 \subset W $, cf.~Figure~\ref{fig:2}.  To formulate a bound, we require the following modification of Defintion~\ref{def:expdec}.

 \begin{figure}[ht]
\begin{center}\includegraphics[width=.75\textwidth]{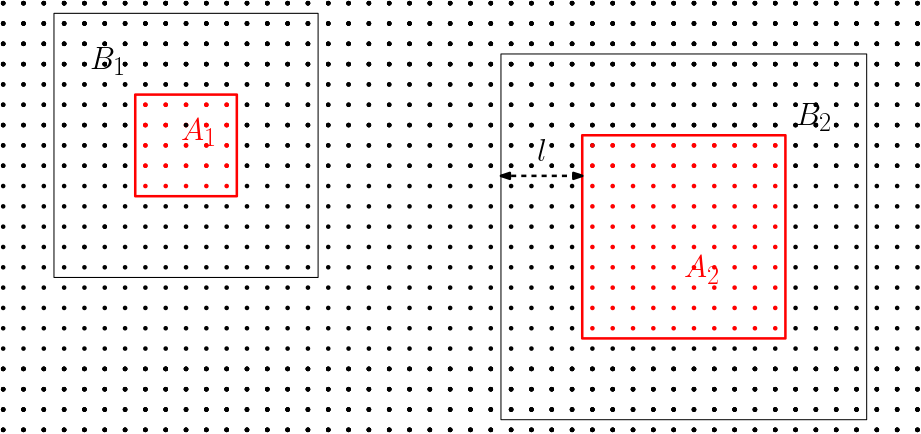}
\caption{Sketch of the geometric situation in Definition~\ref{def:expdec2} and the  proof of Theorem~\ref{thm:mutualinfo}: two cubes $ A_1 , A_2 $ with their buffers $ B_1 , B_2 $ of width $ l  \approx  d(A_1,A_2)/ 3 $.}\label{fig:2} 
\end{center}\end{figure}

\begin{definition}\label{def:expdec2} In the situation of Definition~\ref{def:expdec}, a state $ \ket{\psi} $ is \emph{strongly conditionally $ \varphi $-decoupled at rate $\varphi(l)$  beyond distance $ l_0$}
 if for any $ A \subset W $ composed of two disjoint, regular subsets $ A_1, A_2 $ at distance $ d(A_1,A_2) \geq 3 l $ with their corresponding disjoint buffers $B_{j,l} := \{ u \in W \backslash A  \ | \ d(u,A_j ) \leq l \} $ and  $ l \geq l_0(A)  $, one has
\begin{equation}\label{ass:expdec2}
\TV_{B_l}(A \, |\, C_l  ) + \delta\left( A_1 \sqcup B_{1,l} | A_2 \sqcup B_{2,l} \right)  + \vartheta_{B_{1,l}, B_{2,l}}(A_1 , A_2 |  C_l  )  \leq   \varphi(l-l_0(A))  ,
\end{equation}
where $  B_l := \sqcup_{j=1}^2  B_{j,l} $ and $ C_l := W \backslash ( A  \sqcup B_l) $. 
\end{definition}
\noindent \textbf{Remark:}~~It will be shown in Theorem~\ref{thm:delta} that in the above situation
\begin{equation}\label{eq:estmiddle}
 \delta\left( A_1 \sqcup B_{1,l} | A_2 \sqcup B_{2,l} \right) \leq  \delta\left( A_1 | A_2 \right) +  \delta_{B_{1,l}}\left( A_1 | A_2 \right)  +  \delta_{B_{2,l}}\left( A_1 | A_2 \right) +  \delta_{B_{l}}\left( A_1 | A_2 \right) \, . 
\end{equation}

The case $\varphi(k) = e^{-k/\xi}$ at some $ \xi > 0 $, will be referred to as strongly conditionally exponentially decoupled. 
One may now use the proof of Theorem~\ref{thm:compS} to establish the following variant.   In its formulation, we specify a logarithmic decoupling distance $ l_0(A) \propto \ln |\partial A | $. As explained in Section~\ref{sec:delta}, for ground states of systems with the FKG property such as in the quantum Ising model, present classical statistical mechanics techniques allow to establish exponential decoupling beyond that distance in case one stays away from the ground-state phase transition.

\begin{corollary}[Entropy difference]\label{thm:contbound}
For a state $ \ket{\psi} $, which is strongly conditionally exponentially decoupled beyond distance  $ l_0(A) = c \ln |\partial A| $ for some $ c \in (0,\infty) $, and any $ A \subset W $ composed of two disjoint, regular, non-empty subsets $ A_1, A_2 $ at distance $ d(A_1,A_2) \geq 3 l $ with  $ l\geq l_0(A) $, there are constants $ \widehat C, \eta, \kappa \in (0,\infty) $, which are independent of $l,  l_0 $ and $ A $, such that 
for all three choices $ D \in \{ A , A_1 , A_2 \} $:
\begin{equation}\label{eq:contbound}
\left| S(\varrho_D(\psi)) - S\big(\varrho_D(\psi_{(B_l)}) \big) \right| \leq \widehat C \,    |\partial A|^\kappa     \exp\left( - l / \eta \right) \, ,
\end{equation}
with $ \ket{\psi_{(B_l)} }$ as defined in~\eqref{eq:appr} with $ B \equiv B_l := B_{1,l} \sqcup B_{2,l} $, $ C \equiv C_l := W \backslash (A \sqcup  B_l)  $ and $ \alpha = \alpha_1+ \alpha_2 $, and minimizers $ \alpha_1 ,  \alpha_2 $, $ \gamma $ from~\eqref{eq:K2}. 
\end{corollary}
\begin{proof}
The proof starts by noting that for any of the three reduced states $ \varrho_D \equiv  \varrho_D(\psi) $ with $ D \in \{ A , A_1 , A_2 \} $, the arguments in the derivation of~\eqref{eq:massTV} still yield for any $ l \geq l_0(A) $:
\begin{align*}
\tau_{\varrho_D}\big(\nu^{|B_l|} \big) & \leq 2 \ \TV_{B_l}(A \, |\, W \backslash ( A \sqcup B_l) ) + 2 \ \vartheta_{B_l}(A \, |\, C_l  ) \\
&\leq  2 \ \TV_{B_l}(A \, |\, C_l ) + 2 \  \vartheta_{B_{1,l}, B_{2,l}}(A_1 , A_2 |  C_l  )  \leq  2 \ \varphi(l-l_0(A) ) \, .   
\end{align*}
Applying the same strategy as in Theorem~\ref{thm:compS}, we note that~\eqref{eq:contentgen} continues to hold also for the state $ \ket{\psi_{(B_l)} }$ with $ \varphi(l) = \exp\left( - l/\xi\right) $.  In this case, for any $ l \geq l_0 $:
\begin{align*}
 I_1(l-l_0;A)  & \leq e^{- (l-l_0(A))/\xi}\ \sum_{k=0}^\infty e^{- k/\xi} ( 1+ k + l - l_0(A) ) \, , \\ 
 I_2(l-l_0;A)  & \leq e^{- (l-l_0)/\xi} \ \sum_{k=0}^\infty e^{- k/\xi} \left[ 1 +  (k +  l - l_0(A))/\xi \right] \, , \\
 \sqrt{ \varphi(l-l_0)} \left( \ln \frac{\nu^{|B_{l_0}|}}{\sqrt{\varphi(l-l_0)}} + 1 \right) & \leq e^{- (l-l_0)/(2\xi) } \left( C_d\ |\partial A| \ l_0\  (\ln \nu) + 1+ (l-l_0)/(2\xi)  \right) \, . 
\end{align*}
Since $ \exp\left( l_0(A)/\xi \right) = |\partial A |^{c/\xi} $ and any polynomial in $ l $ can be absorbed by peeling of a part of the leading exponential, this implies the result. 
\end{proof}

 \section{Mutual information bounds}\label{sec:mutual}
 
The  above approximation results yield interesting consequences for the mutual information $ I_\psi(A_1 :A_2 ) $ encoded in a strongly conditionally exponentially decoupled state $ \ket{\psi} $ and two disjoint subsets $ A_1 , A_2 \subset W $, cf.~Figure~\ref{fig:2}. Corollary~\ref{thm:contbound}  allows to approximate each of the entropies entering the definition~\eqref{def:mutualinfo} of the mutual information in terms of the entropy of a comparison product state with vanishing mutual information.

\begin{theorem}[Exponential decay of mutual information]\label{thm:mutualinfo}
Let $ \ket{\psi} $ be a state, which is strongly conditionally exponentially decoupling beyond $ l_0(A) = c \ln |\partial A| $ for some $ c \in (0,\infty) $, and $ A = A_1 \sqcup A_2 $ with $ A_1, A_2 \subset W $ non-empty squares, which are separated by distance $ d(A_1,A_2) \geq 3 l_0( A) $. Then there are constants $ C, \kappa, \eta \in (0,\infty) $, which are independent of $ A_1 , A_2 $, such that
\begin{equation}\label{eq:mutual_info_bound5}  
	I_\psi(A_1 : A_2)  \leq C \max\left\{ |\partial A_1|, |\partial A_2 | \right\}^\kappa \ \exp\left(- d(A_1,A_2)/\eta \right) \, . 
\end{equation}
\end{theorem}

Before spelling out the proof, let us note that the quantum Pinsker inequality~\eqref{eq:Pinsker} implies that under the conditions of Theorem~\ref{thm:mutualinfo} the correlations of \emph{any} pair of local observables are exponentially decaying in the observables' distance - and not only those diagonal in the computational basis. 

Moreover, as is evident from the subsequent proof, Theorem~\ref{thm:mutualinfo} admits a generalization to the mutual information of finitely many disjoint squares $ A_1, \dots , A_n $.

\begin{proof}[Proof of Theorem~\ref{thm:mutualinfo}]
We abbreviate by $ \varrho_{A_1} \equiv  \varrho_{A_1}(\psi) $, $ \varrho_{A_2} \equiv  \varrho_{A_2}(\psi) $ and $ \varrho_{A} \equiv  \varrho_{A}(\psi) $ the reduced states of the subsets and their union $ A = A_1 \sqcup A_2 $. Furthermore, we set $ l := d(A_1,A_2)/ 3 $ and denote by $$ B_{j} \equiv B_{j, l } := \{  u \in W \backslash A_j  \ | \ d(u,A_j) \leq l \} $$  the buffers of width $ l $ around $ A_j $ with $ j \in \{1,2\} $, cf.~Figure~\ref{fig:2}. We also set 
$$  B \equiv B_l := B_{1,l} \sqcup B_{2,l} \quad\mbox{and}\quad  C := W \backslash (A\sqcup B) \, .$$

The proof then proceeds by a two-step approximation. We first use Corollary~\ref{thm:contbound} to approximate each of the three entropies in~\eqref{def:mutualinfo} in terms of
the respective entropies of the approximating state $ \ket{\psi_{(B)} } $ defined in~\eqref{eq:appr} with $ \alpha = \alpha_{1} +\alpha_2 $, $ \gamma $ the minimizers from~\eqref{eq:K2}. Corollary~\ref{thm:contbound} with $ l_0(A) = c \ln |\partial A | $ then guarantees that these entropy differences are upper bounded by the right side in~\eqref{eq:mutual_info_bound5}.\\

In a second step, we approximate the state $ \ket{\psi_{(B)}} $ by the  state 
\begin{equation}
 \ket{\psi^{(12)}_{(B)} }  := \sum_{\BS} \exp\left( i \sum_{j=1}^2 \alpha_j(\BS_{A_j},\BS_{B_j}) + i \gamma(\BS_C,\BS_B) \right)  \sqrt{ p^{(12)}_{(B)}(\BS)}  \ \ket{\BS} 
 \end{equation} 
 defined in terms of the probability
\begin{equation}
 p^{(12)}_{(B)}(\BS)  := p(\BS_{C} | \BS_{B} )  \prod_{j=1}^2 p(\BS_{B_j})   p(\BS_{A_j} | \BS_{B_j} )   \  \, . 
\end{equation}
The following are straightforward extensions of~\eqref{reduced} and Theorem~\ref{thm:fidcor}:
\begin{enumerate}
\item The reduced state of $  \ket{\psi^{(12)}_{(B)} }  $ on $ A = A_1 \sqcup A_2 $ is the product state
\begin{equation}
 \varrho_{A}(\psi^{(12)}_{(B)}) =  \varrho_{A_1}(\psi^{(12)}_{(B)}) \otimes  \varrho_{A_2}(\psi^{(12)}_{(B)}) \, , 
\end{equation}
where $  \ket{\varphi_{A_j}(\BS_{B_j}) }:= \sum_{\BS_{A_j}}e^{i \alpha_j(\BS_{A_j},\BS_{B_j})} \sqrt{ p(\BS_{A_j} | \BS_{B_j}) }  \ \ket{\BS_{A_j}}$ with $ j \in \{1,2\} $, and 
$$
 \varrho_{A_j}(\psi^{(12)}_{(B)}) = \sum_{\BS_{B_j}} p(\BS_{B_j}) \  \ket{\varphi_{A_j}(\BS_{B_j}) } \bra{\varphi_{A_j}(\BS_{B_j}) } \, . 
$$
\item The fidelity with respect to $ \ket{\psi_{(B)}} $ is bounded by
\begin{align}\label{eq:fidest2}
 \left| 1 - \braket{\psi^{(12)}_{(B)}}{\psi_{(B)} }\right|
& \leq \sum_{\BS_A,\BS_B}    \Big| \ p(\BS_B)p(\BS_A | \BS_B)  -  \prod_{j=1}^2 p(\BS_{B_j})   p(\BS_{A_j} | \BS_{B_j} ) \Big|   \notag \\
& \leq   \delta\left( A_1 \sqcup B_1 | A_2 \sqcup B_2 \right) \leq \varphi(l - l_0(A))  \, .
\end{align}
Note that the phase factors in $\ket{\psi_{(B)}} $ and $ \ket{\psi^{(12)}_{(B)}} $ agree and hence drop out in the scalar product. 
\end{enumerate}
Since the mutual information of the approximating state $ \ket{\psi^{(12)}_{(B)}} $ vanishes, $$ I_{\psi^{(12)}_{(B)}}(A_1 : A_2 ) = 0 \, , $$
it remains to bound the entropy differences 
$ S\big(\varrho_A(\psi^{(12)}_{(B)}) \big) - S\big(\varrho_A(\psi_{(B)}) \big) $ and similarly for $ A $ replaced by $ A_1 $ and $ A_2 $. This is done with the help of the Fannes' type bound, Proposition~\ref{cor:Fannes}. 
For its application, we note that the rank of the occurring state differences are bounded by $ \dim \mathcal{H}_B = \nu^{|B|} = \nu^{|B_1|+|B_2|} $, where by~\eqref{eq:Buffervol}, one may further estimate $ |B_j| \leq C_d |\partial A_j| \ l $. 

The trace distances occurring in the application of Proposition~\ref{cor:Fannes} are estimated by fidelities using~\eqref{eq:FvG}, which is upper bounded using~\eqref{eq:fidest2}:
\begin{align}
  2 \left( 1 - \Re \braket{\psi^{(12)}_{(B)}}{\psi_{(B)} }\right) & \leq 2   \delta\left( A_1 \sqcup B_1 | A_2 \sqcup B_2 \right) \notag \\
&  \leq 2 e^{- (l-l_0(A))/(3\xi)} = 2 \ |\partial A|^{c/(3\xi)} e^{-l/(3\xi)} \, . 
\end{align}
The last follows from the assumption and the fact that $ d(A_1 \sqcup B_1 , A_2 \sqcup B_2) \geq l / 3 $ by construction.
Since $ |\partial A| \leq 2 \max\left\{ |\partial A_1|, |\partial A_2 | \right\} $,  this completes the proof. 
\end{proof}

\section{Bounds on conditional dependence}\label{sec:delta}

The purpose of this section is to elaborate further on the inter-domain correlation functions introduced in Section~\ref{sec:fidelity}.  These  play a central role in our analysis of the probabilistic part, and its applications to specific models. 
\subsection{General properties} 

We first collect a number of useful properties of $\TV(A \, | \, C)$ defined in~\eqref{def_Delta} for arbitrary disjoint subsets $ A, C $ among which is last but not least~(iii), which is a variant of subadditivity.  

\begin{theorem}   \label{thm:delta}
For any probability measure $p$ on a product space $\Omega= \Omega_0^W$ with $W$ a finite set and $\Omega_0$ the range of values of the single ``spin variables'', the inter-domain correlation  has the following properties  for disjoint $ A, C, D \subset W $:
\begin{enumerate} [i)]
\item  symmetry:  $\TV(A \, | \, C) = \TV(C \, | \, A)$ 
\item  monotonicity: 
$
\TV(A \,  | \, C) \ \leq\  \TV(A  \, | \, C \sqcup D) \ 
$
\item sub-cocycle   
   for     disjoint $C, D\subset W\setminus A$ 
\be\label{sub_add}
 \TV(A  \, | \, C \sqcup D) \  \leq \  \TV(A \, | \, C)  + \TV_C(A \, | \,  D) \, .
\ee  
\end{enumerate} 
\end{theorem} 
\begin{proof} i) The $A \leftrightarrow C$ symmetry is obvious from the defining equation~\eqref{def_Delta}. 

 ii)  By a known variational characterization of the ``total variation'' between pairs of probability measures,  \eqref{def_Delta} can be  restated as
 \be \label{var}
\TV(A \, | \, C) = 
\frac 12   \max_{ \|F\|_\infty \leq 1}  \sum_{\BS_A,\BS_C} F(\BS_A,\BS_C) \left[ p(\BS_A \BS_C) - p(\BS_A) p(\BS_C) )\right]  \, ,
 \ee    
 where the maximum is over $ F : \Omega_A \times \Omega_C \to [-1,1] $ and attained at $ F(\BS_A,\BS_C) := \sign [ p(\BS_A,\BS_C)  - p(\BS_A) p(\BS_C)\big]$. 
Monotonicity in $C$ follows, since replacing $\BS_C$ by $\BS_{C\sqcup D}$ only enlarges the collection of functions over which the maximization is performed.

iii) 
Regrouping the relevant  conditional probabilities  into two parts, and applying the triangle inequality, we get: 
\begin{align} \label{proof_subadd}
 &  \left| p(\BS_A,\BS_C,\BS_D) - p(\BS_A) p(\BS_C,\BS_D) \right| \notag \\
 & =   p(\BS_C)  \left| p(\BS_A, \BS_D | \BS_C)  - p(\BS_A) p(\BS_D | \BS_C) \right| \\
  & \leq     p(\BS_C,\BS_D) \left| p(\BS_A | \BS_C) -  p(\BS_A ) \right|  \notag \\
  & \mkern150mu + p(\BS_C) \left| p(\BS_A, \BS_D | \BS_C) -  p(\BS_A | \BS_C)  p( \BS_D | \BS_C)  \right| . \notag
\end{align}
The sums over configurations of the first and second term on the right coincide with the first and second term on the right side of~\eqref{sub_add}, respectively.  Thus the claim follows from the  inequality in \eqref{proof_subadd}.
\end{proof}

Subadditivity  allows to upper bound the conditional inter-domain correlation $ \TV_B(A \, | \, C) $ in terms of  sums of domain-to-point correlations $ \TV_D(A \, | \, \{u\} ) $ with suitable $ D \supset B $.  It is therefore interesting to note that the latter quantity is upper bounded in terms of an averaged conditional  total variation of two measures, which compares the influence under the flip of a single variable. Specifically, for disjoint $ A, B, \{u\} \subset W $, we set
\begin{equation}
\overline{\Tv}_{ \BS_B}(A |  \{ u \}) := \sum_{\sigma_u, \widehat \sigma_u\in \Omega_0 } p(\sigma_u | \BS_B) p( \widehat  \sigma_u | \BS_B)   \ \Tv_{ \BS_B}(A ; \sigma_u, \widehat \sigma_u ) 
\end{equation}
with
\begin{equation}
\Tv_{ \BS_B}(A ; \sigma_u, \widehat \sigma_u ) :=  \frac{1}{2}  \sum_{\BS_A} \left| p(\BS_A | \BS_B \sigma_u) - p(\BS_A | \BS_B \widehat \sigma_u) \right| 
\end{equation}
In case $ B = \emptyset $ and for Gibbs measures $ p $, the exponential decay in the distance of $ A $ to $ u $ of the uniform bound $ \sup_{\sigma_u, \widehat \sigma_u }
\Tv_{ \BS_B}(A ; \sigma_u, \widehat \sigma_u ) $ of this total variation is one of the 12 equivalent characterisations ('Condition IIIa') of the measure's  'high-temperature' regime in Dobrushin and Shlosman seminal work~\cite{DS87}.

\begin{lemma}
In the situation of Theorem~\ref{thm:delta} for any disjoint  $ A, B  \subset W $ and $ u \in W \backslash (A\sqcup B) $:
\begin{equation}\label{eq:DScond} 
\TV_B(A | \{u\} ) \leq  \sum_{ \BS_B} p(\BS_B) \ \overline{\Tv}_{ \BS_B}(A |  \{ u \}) \, . 
\end{equation}
\end{lemma}
\begin{proof}
This is a simple consequence of the representation
$$
p(\BS_A | \BS_B \sigma_u) -  p(\BS_A | \BS_B)  = \sum_{\widehat \sigma_u}  p(\widehat  \sigma_u | \BS_B) \left[ p(\BS_A | \BS_B \sigma_u) - p(\BS_A | \BS_B \widehat\sigma_u) \right] \, , 
$$
the triangle inequality and the definition~\eqref{def_cond_B}  of $ \TV_B(A |\{u\} ) $. 
\end{proof}

\subsection{Relation to correlation function using the FKG property}   
\label{sec:FKG}

When applicable, the FKG inequality provides a useful tool for deducing decoupling estimates from  the decay rate of the two-point spin-spin correlation function.   In its statement we refer to real valued functions on $\RR^n$ as monotone increasing iff they are so in the component-wise sense.   
\begin{definition} 
A probability measure $\mu$ on a product space $\Omega_0^W$,   over a finite set $W$ and based space $ \Omega_0 \subset \mathbb{R} $,  is said to have the  \emph{positive association} property 
 if for any pair of  
 increasing, bounded functions $f,g: \Omega_0^W \to \mathbb{R} $: 
\be 
  \langle f\, g \rangle_\mu  -  \langle f  \rangle_\mu  \langle  g \rangle_\mu  \geq 0 \, , 
\ee 
Furthermore,  if  this condition holds also for the averages  conditioned on prescribed values of the local variables over arbitrary  subset  $B\subset W$, the probability is said to have the \emph{strong positive association} property.  
The latter is also referred to as Fortuin-Kasteleyn-Ginibre (FKG) property.  
\end{definition} 

The FKG property is known to  apply to a number of models of interest in statistical mechanics~\cite{Grim06}.  Included in this class is  the  classical Ising model and also the quantum ferromagnetic  Ising model with the Hamiltonian \eqref{QIM_H}  at any $b\geq 0$, and  $h\in \RR$.  
For such a system of qubits the following  result  allows to simplify the verification of the decay of the probability part in Definition~\ref{def:expdec}.

\begin{theorem} \label{thm_stoQuastic}  
Given a probability $ p $ with the FKG property defined over the space of binary configurations $ \{+1,-1\}^W$:
\begin{enumerate}
\item for disjoint $ A, B , \{u\} \subset W $: 
\be \label{key_lemma}
\TV_B(A\,|\, \{u\}) \leq  \frac 1 4 \sum_{a \in A} \langle \sigma_a ; \sigma_u \rangle_B \, ,
\ee
where $\TV_B$ is the conditioned  inter-domain correlation defined by \eqref{def_cond_B},   
\be 
\langle \sigma_a ; \sigma_u \rangle_B  := \sum_{\BS_B} p(\BS_B)\ \langle \sigma_a ; \sigma_u \rangle_{\BS_B}\, , 
\ee 
and $\langle \sigma_a ; \sigma_u \rangle_{\BS_B} := 
\langle \sigma_a \sigma_u\rangle_{\BS_B} -  
\langle \sigma_a \rangle_{\BS_B}  \langle  \sigma_u\rangle_{\BS_B}$
with the subscript indicating that the expectation values are with respect to the conditional probability $ p(\cdot |\BS_B )$, conditioned on the value of $\BS_B$.
\item  for disjoint $ A, B , C\subset W $,  the  conditional correlation defined in \eqref{def_cond_B} satisfies
\be \label{eq:deltaestK}
\TV_B(A \, | \, C) \leq  \frac 14  \sum_{u\in A, v\in C} K_B(u,v) \, . 
 \ee
with 
\begin{equation}\label{eq:definflk}
K_B(u,v) := \max_{D \supset B} \max_{\BS_D } \ \langle \sigma_u ; \sigma_v \rangle_{\BS_D} \, . 
\end{equation}
\end{enumerate}
\end{theorem}
\begin{proof} 
1.~~We start from~\eqref{eq:DScond}. In the binary case, abbreviating by $ p_{\BS_B}(\pm ) := p(\sigma_u= \pm |\BS_B) $, we have 
$$
  \overline{\Tv}_{ \BS_B}(A |  \{u \} )  =  p_{\BS_B}(+) \  p_{\BS_B}(- ) {\Tv}_{ \BS_B}(A) ,
$$
with 
\begin{align}
  {\Tv}_{ \BS_B}(A) 
  & :=  \frac{1}{2}  \sum_{\BS_A} 
\left| p(\BS_A |\BS_B \sigma_u=1) - p(\BS_A | \BS_B \sigma_u =-1)\right|   \notag \\ 
&  = \min \left\{ \mathbb{P}\left(\BS_A^{(1)} \neq  \BS_A^{(2)}\right) \ \big| \ \mbox{couplings with $\BS_A^{(1)}\rightsquigarrow \BS_A^+,  \BS_A^{(2)}  \rightsquigarrow \BS_A^-$}\right\}  \, . \notag 
\end{align}
where, by a generally valid formula of the total variation, the minimum is over probability distributions $ \mathbb{P} $ of pairs $(\BS_A^{(1)}, \BS_A^{(2)})$  with the marginals indicated indicated by $\rightsquigarrow$. Here $ \BS_A^\pm $ denote the random variables conditioned on $ \BS_B  $ and $  \sigma_u =\pm 1 $. 

  Under the assumed strong positive association property of the measure, 
 the Strassen-Holley  theorem~\cite{HolleyR74} assures the existence of couplings which are monotone in the sense  that for almost every joint realisation of the  pair 
\be 
 \sigma_a^{(1)} \geq \sigma_a^{(2)}  \quad  \mbox{for all $a\in A$} \,.
\ee
Using any such coupling we get: 
\begin{align}
   \mathbb{P}\big(\BS_A^{(1)}  \neq   \BS_A^{(2)} \big) & \leq  
\EE \left(\sum_{a\in A} (\sigma_a^{(1)} -\sigma_a^{(2)})/2  \right)  \notag \\ 
&= \frac 12 \sum_{\BS_A}  M_A(\BS_A)   \big[p(\BS_A | \BS_B \sigma_u = +1)  - p(\BS_A | \BS_B \sigma_u = -1) \big] \notag 
\end{align}  
in terms of the function  $M_A(\BS_A)  := \sum_{a\in A} \sigma_a$.

The above difference  also appears in the following expansion for the covariance  
$\langle M_A ; \sigma_u\rangle_{\BS_B} $
where the average is with respect to the conditional distribution $p(\BS_A,\sigma_u | \BS_B )$.   That is easily seen by processing 
a duplicate distribution of  
an independent pair $(\BS_A,\sigma_u)$, $(\widehat\BS_A,\widehat \sigma_u)$ of the same distribution:
\begin{align}  \label{Delta_Cov_u}
&\langle M_A ; \sigma_u\rangle_{\BS_B}   \notag \\
& =\frac 12 \sum_{ \sigma_u,  \widehat\sigma_u } \sum_{\BS_A , \widehat\BS_A} p(\BS_A,\sigma_u | \BS_B) \ p(\widehat\BS_A,\widehat\sigma_u  | \BS_B)
\left[M_A(\BS_A) - M_A(\widehat\BS_A) \right] \left[ \sigma_u - \widehat{\sigma}_u  \right]   \notag \\ 
& = 2 \, p_{\BS_B}(+) \  p_{\BS_B}(- ) \sum_{\BS_A}  M_A(\BS_A)   \big[p(\BS_A | \BS_B\sigma_u = +1)  - p(\BS_A | \BS_B \sigma_u = -1) \big]
\end{align}  
This completes the proof of the first item.\\

\noindent
2.~~Iterating \eqref{sub_add} one learns that for any monotone sequence $C_0:=B \subset C_1  ...  \subset C_k=: C$, one has:
\be\label{eq:tele}
 \TV_B(A  \, | \, C) \  \leq \  \sum_{j=1}^k  
 \TV_{C_{j-1}}(A \, | \,  C_j\setminus C_{j-1}) 
\ee  
In particular, the above applies to sequential decompositions of  $C$ with single site increments, $C_j= C_{j-1} \cup \{u_j\}$, in which case by~\eqref{key_lemma} and the definition~\eqref{eq:definflk} of the influence kernel $ K_B $: 
$$
 \TV_{C_{j-1}}(A \, | \, \{u_j\}) \leq \frac{1}{4} \sum_{a \in A} \sum_{\BS_{ C_{j-1}}} p(\BS_{ C_{j-1}}) \ \langle \sigma_a ; \sigma_{u_j} \rangle_{\BS_{ C_{j-1}}} \leq \frac{1}{4}  \sum_{a \in A} K_B(a,u_j) \, .
$$
Inserting this estimate into~\eqref{eq:tele}, produces~\eqref{eq:deltaestK}.
\end{proof}

\section{Application to the quantum Ising model}
\label{sec:QIM}

\subsection{Explicit results}

The model and its phase transitions were outlined in the introduction. 
Its ground state expectation value functional can be presented as 
\be  \label{trace_ratio}
\langle O \rangle^{(W)} = \lim_{\beta\to \infty} 
\frac{ \tr \ O\, e^{-\beta H} } { \tr  e^{-\beta H}} 
\ee 
with $H$ given by \eqref{QIM_H} acting on $ \mathcal{H}_W $.   
The Hamiltonian includes non-commuting terms. In the computational basis based on $S^z$ the  terms involving $S^x$ act as flip operators.   It readily follows that the model is stoquastic in this basis (and also in the $S^x$ basis, in whose terms  $S_u^z$ act as flip operators).  

The part of $H$ involving $S^z$ spins looks identical to the classical Ising spin model over the same lattice.   In the natural (Dyson-Trotter-Feynman-Kac) functional representation of the trace in \eqref{trace_ratio}  the expectation value of functions of $\{S_u^z\}_{u\in W}$  resemble a Gibbs average over a ``time''-dependent spin function $\BS: W\times [0,\beta] \to \{-1,+1\}$.  The latter, other than being formulated over a partly continuous set, resembles the classical Ising model over a $(d+1)$ dimensional set. 

As it turns out, the relation extends beyond a vague analogy (cf.~\cite{suzuki2012quantum,AKN94,Bjornberg:2009aa,Ioffe:2009aa,Crawford:2010aa}):    the continuum model on $W\times \RR$ may  be presented as a week limit  of discrete models on $W\times (\varepsilon \ZZ) $,  at suitably adjusted couplings (which for ``time-wise'' neighbour's coupling diverge as $\varepsilon  \downarrow 0$).   

Through this similarity and relation,  results, which were originally derived  for the classical Ising model, have found extensions to the 
the model's quantum version.  Of direct relevance here are the following statements, in which we focus on the finite- and infinite-volume ground states (corresponding to the free boundary conditions). 
\begin{enumerate}  
\item \emph{FKG}:  Like its classical version~\cite{HolleyR74,Grim06}, the quantum Ising model's probabilty has the FKG property, in the natural partial order of the stoquastic $S^z$ basis (cf.~\cite{Grimmett:2008aa}). This makes Theorem~\ref{thm_stoQuastic}  applicable to this model, with its  variables $\sigma_u$ interpreted as $S^z_u$.  

\item  \emph{Griffiths correlation inequalities}:  As stated and discussed in~\cite{friedli_velenik_2017,BSim24}, these imply that: i) 
the finite volume ground states' correlation functions converge as asserted in \eqref{eq:corr_lim}, ii) for each $W\subset\ZZ^d$  the ground states' spin-spin correlation function (of $S^z$) is dominated by its infinite volume limit.  In the notation of \eqref{eq:corr_lim}:
\be  
\langle S^z_u S^z_v \rangle^{(W)}  \leq 
\langle S^z_u S^z_v \rangle^{(\ZZ^d)} \,. 
   \ee 

\item \emph{Sharpness of the phase transition:} For the QIM it was proven by Bj\"ornberg and Grimmett~\cite{Bjornberg:2009aa} that for any dimension $d\geq 1$ and   ferromagnetic coupling 
$J=\{J_{x-y} \} \not\equiv 0$, there is a finite $b_c$ at which the model undergoes a quantum phase transition, as described above \eqref{eq:shapness}.  
In particular, for all $ b > b_c$, the unique ground state exhibits exponential decay at a finite correlation length ($\xi(b) <\infty $) 
\be \label{QIM_expdecay}
\langle S^z_u  S^z_v\rangle^{(\ZZ^d)}   \leq   e^{- |u-v|/\xi(b)} \, .
\ee
\item \emph{Ding-Song-Sun inequality:} at any $b $ (omitted in our  notation, but to be understood as of common value within the stated relation) and for all finite subsets $D\subset \Z^d$ and specified configuration $\BS_D =\{\sigma_u\}_{u\in D}$ of $ z $-spins:
\be \label{QIM_DSS}
\langle S^z_u; S^z_v\rangle_{\BS_D}  \leq   \langle S^z_u S^z_v\rangle \, . 
\ee
In its classical version, the inequality is a direct implication of the recently formulated Ising inequality, which was presented in \cite[Eq.~(1.4)]{Ding:2023aa}.  
The validity of its quantum version follows by a weak continuity argument, similar to that spelled in the discussion of  the FKG inequality for QIM in \cite{Grimmett:2008aa}.  An alternative direct proof, of both classical and quantum versions using the random current approach,  is included in \cite{Aiz_III}.  
\end{enumerate}

\begin{proof}[Proof of Theorem~\ref{thm:EEQIM}]
The ground state of the quantum Ising model in a finite volume $W\subset \ZZ^d$ with free boundary conditions is a unique pure state at $ b > 0 $ due to the Perron-Frobenius theorem and stoquastic in the $z$-diagonalising computational basis.
By the extended Ding-Song-Sun inequality \eqref{QIM_DSS}, for any finite volume and throughout the range of the model's coupling constants, its ground state's  truncated correlation function, conditioned on arbitrary configuration of the buffer, satisfies: 
\be \label{609}
\langle S^z_u ; S^z_v \rangle_{\sigma_B}^{(W)} \leq \langle S^z_u S^z_v \rangle^{(W)}  \leq \langle S^z_u S^z_v \rangle^{(\ZZ^d)} ,
\ee 
where the last upper bound holds thanks to the Griffiths inequality.

Combining that with  ``sharpness of the phase transition'', i.e. 
exponential decay of correlations in the infinite-volume limit up to the threshold of symmetry breaking, 
we learn that  for any $b>b_c$ the above upper bound decays exponentially, as stated in \eqref{QIM_expdecay}.  
Thus, the kernel, which is defined in \eqref{eq:definflk}, through maximization of $\langle S^z_u ; S^z_v \rangle_{\sigma_B}^{(W)}$ satisfies  
\be 
K_B(u,v) 
\leq e^{-|u-v|/\xi(b)}
\ee 
at some $\xi(b)<\infty$. 

The model's FKG property makes Theorem~\ref{thm_stoQuastic} applicable, boosting the implications of the two point bound.  The bound~\eqref{eq:deltaestK}  allows to conclude that for each $b>b_c$, each choice of domains, and 
$l\geq 1$ 
\begin{eqnarray} 
\delta_{B_l}(A | C) &\leq& \frac 14 \sum_{u\in A, v\in C} \, e^{-|u-v|/\xi(b)}  \notag \\  
&\leq& C\,  |\partial A| \, \xi(b) e^{-l/ \xi(b)}  \leq e^{-(l-l_0(A)/\xi(b)} \label{eq:delta_small}
\end{eqnarray}  
with $l_0(A) = C  \ln |\partial A| $.

Combined with Corollary~\ref{thm:gen_area1}, this model specific input  implies  \eqref{eq:areaQIM}.  

The proof of the claims concerning the mutual information and exponential clustering of observables reduces to an application of~\eqref{eq:Pinsker} and Theorem~\ref{thm:mutualinfo}.  
Its assumptions are satisfied thanks to~\eqref{eq:delta_small}, which by~\eqref{eq:estmiddle} also implies that, in the situation of Definition~\ref{def:expdec2}, the middle term in the left side of~\eqref{ass:expdec2} is exponentially bounded.\end{proof} 

\subsection{Discussion} 

Let us close with some remarks and open questions.

\begin{enumerate}
\item  In previous studies of the entanglement in QIM \cite{Grimmett:2008aa,Campanino:2020aa,Grimmett:2020aa}  decorrelation effects were expressed through a  \emph{ratio weak-mixing} estimate, in which the proximity of two probability distributions 
$p(\BS_A)$ and $\widehat{p}(\BS_A)$ is approached through an upper bound on $|p(\BS_A)/\widehat{p}(\BS_A) -1|$.   
This approach works well when  $|\delta p|(\BS_A):=|p(\BS_A)-\widehat{p}(\BS_A)|$ is typically much smaller than $p(\BS_A)$ itself.   
However that condition fails in dimensions $d>1$ even if $\delta p$ is of order $e^{-L(A)/\xi}$, since there $p(\BS_A)$ is of the order $e^{-\alpha L(A)^d}$.  In one dimension it applies but only where  the correlation length  is shorter than a finite threshold. 
This limitation is avoided in our estimate of the trace distance between the corresponding states through a more effective use of the FKG inequality. 
\item It seems reasonable to expect that an entanglement entropy bound similar to what is proven here for subcritical states 
is valid also for each of the pure ground states of broken symmetry, away from the critical case.    
In that case  the conditional  variational distance $\delta_{\BS_B}(A|C) $ may  be small for configurations of the buffer that  are aligned with the state's  order parameter, but not so for  configurations $\BS_B$ which are magnetised in the opposite direction.  
Still, with $\delta_{B}(A|C) $ being an average of the above, one may expect the criterion of  Corollary~\ref{thm:gen_area1}  to be applicable also in that case - provided progress is made on a more effective handling of the above observations.  

A positive, though still insufficient, result in the above direction is the existing proof, by Duminil-Copin, Tassion, and Raoufi~\cite{Duminil-Copin:2020aa}, that in the classical Ising model in all dimensions the \emph{truncated}  two point function decays exponentially fast at all states except those along  the threshold of symmetry breaking~\cite{Duminil-Copin:2020aa}. 

\item Among the numerous other examples of stoquastic Hamiltonians to which the above strategy may be of relevance  is the quantum Heisenberg model~\cite{Dyson:1978aa,Toth:1993aa,Uel13,Tas20},  its `flattened' projection based interactions,  studied by related methods in~\cite{Klumper:1990aa,Aizenman:1994aa,Aizenman:2020aa} (and by different methods in many other works), and also the ground states of Kitaev's toric code~\cite{Kitaev:2006aa,Tas20}. 
\end{enumerate}
\appendix 

\section{Fannes' type bound}\label{App:Fannes}

In controling  of entropy differences we made use of the following minor modification of Fannes'  entropy bound~\cite{Fannes:1973aa}, which appeared above as~\eqref{eq:Fannes}, on the difference of von Neumann entropies of two states $  \varrho , \widehat \varrho $.  
 For the reader's convenience we include a short proof, which is slightly different than the one in~\cite{Fannes:1973aa}. 

\begin{proposition}\label{cor:Fannes}
For any pair of states $ \varrho , \widehat \varrho  $ on a finite-dimensional Hilbert space: 
\begin{equation} \label{S_bound}
\left| S(\varrho) - S(\widehat \varrho)\right| \leq   \frac{1}{2} \left\| \varrho -  \widehat \varrho \right\|_1 \left(1+  \ln \frac{2\,   \rank (\varrho - \widehat \varrho) }{ \left\| \varrho -  \widehat \varrho \right\|_1}  \right) .
\end{equation}
\end{proposition}

Bounding $ \rank (\varrho - \widehat \varrho) $ by the dimension of the Hilbert space on which the states are defined,  \eqref{S_bound}   reduces to the Fannes bound of \cite{Fannes:1973aa}. Unlike the optimal bound by Audenaert~\cite{Audenaert:2007aa}, this bound 
is not sharp.

Our proof is based on the subadditivity of the trace of the extension of the following function to  non-negative operators.   
$$
F(x) := \begin{cases} x \left( 1+ \ln x^{-1} \right) , & x \in [0,1] , \\
1 , & x > 1 .   \end{cases}
$$
\begin{lemma}\label{lem:Ftrace}
For any pair $ A , B \geq 0 $ of non-negative matrices 
\begin{equation}\label{eq:Ftrace}
\tr F(A+B) \leq \tr F(A) + \tr F(B) 
\end{equation}
\end{lemma}
\begin{proof}
Since $ F $ is differentiable on $ (0,\infty) $ with derivative $ F'(x) = \ln x^{-1} $ in case $ x \in (0,1) $ and $ F'(x) = 0 $ in case $ x \geq 1 $, we may compute for any $ t \in (0,1) $:
$$
\frac{d}{dt} \tr F(A+tB) = \tr F'(A+tB) B  = - \tr B \ln (A+tB)  \leq  - \left( \tr B \ln t + \tr B \ln B \right) . 
$$
The inequality resulted from the estimate $ A + t B \geq tB $ and the operator monotonicity of the logarithm. Integration of the above inequality using $ \int_0^1 \ln( t^{-1} ) dt = 1 $,  yields the first claim. 
\end{proof}

\begin{proof}[Proof of Proposition~\ref{cor:Fannes}]
We first use monotonicity of $ F $ to estimate $ \tr F(\varrho) \leq \tr F( A+B ) $ with  $ A =  \widehat \varrho $ and $ B = \left[ \varrho - \widehat \varrho \right]_+ $, where the latter denotes the positive part of the difference. Lemma~\ref{lem:Ftrace}  thus yields:
$$
 S(\varrho) - S(\widehat \varrho) = \tr F(\varrho) - \tr F( \widehat \varrho ) \leq \tr F( A+B ) - \tr F( \widehat \varrho ) \leq  - \tr B \ln B + \tr B .
$$
The first term on the right is trivially estimated by $ \tr B \ln ( \rank B / \tr B ) $. Since $ \tr ( \varrho - \widehat \varrho  ) = 0 $, we have $ \tr B = \frac{1}{2} \left\| \varrho -  \widehat \varrho \right\|_1 $. This concludes the proof, since the argument holds also for the reversed order of the two states.  
 \end{proof}

\begin{remark}
In the set-up of Lemma~\ref{lem:Ftrace},  Jensen's inequality applied to the second term in the right of~\eqref{eq:Ftrace} yields
\begin{equation} \label{F_AB}
 \tr F(A+B)  \leq \tr F(A) + N F\left(N^{-1}\tr B\right)  \, . 
 \end{equation} 
 This inequality holds true more generaly, for any non-decreasing, concave function $ F:[0,\infty) \to [0,\infty) $.  A proof may be based on an optimal strategy to distribute the total spectral of shift $ \tr B $ which the eigenvalues of $ A +tB$ undergo as $t$ is increased from $0$ to $1$. 
\end{remark}

\minisec{Acknowledgments}
This work was supported by the Simons Foundation (MA) and by the DFG (SW) under EXC-2111--390814868 and  DFG--TRR 352--Project-ID 470903074.
SW would like to thank Princeton University for hospitality.

\minisec{Declarations}
The authors declare that they have no competing financial interests or personal
relationships that could have appeared to influence the work reported in this paper. Data sharing does not apply to this article as no datasets were generated or
analyzed during the current study. 
\bibliography{QIsing.bib}{}
\bibliographystyle{plain}

\end{document}